\documentclass[sigconf,nonacm]{acmart}
\AtBeginDocument{%
  }

\usepackage{amsfonts}

\usepackage{amssymb,amsmath,amsthm}
\usepackage{mathrsfs}
\usepackage{xcolor}
\usepackage{hyperref}
\usepackage{bookmark}
\usepackage{xspace}
\usepackage{algorithm}
\usepackage[noend]{algpseudocode}
\usepackage[all]{xy}

\usepackage{tikz}
\usetikzlibrary{positioning}
\usetikzlibrary{arrows}
\usetikzlibrary{petri}
\usepackage[appendix=append]{apxproof}
\usepackage{caption}

\renewcommand{\appendixprelim}{\newpage}

\renewcommand{\vv}{d}

\newcommand{\G}{\mathcal{G}}

\newcommand{\picunit}{3cm}
\newcommand{\halfunit}{1.1cm}

\newcommand{\transitions}{{\mathcal T}}
\newcommand{\trans}{\longrightarrow_{\transitions}}

\newenvironment{slexample}
    {\begin{example}\rm}
    {\hfill $\triangleleft$ \end{example}}
\newenvironment{sldefinition}
    {\begin{definition}\rm}
    {\hfill $\triangleleft$ \end{definition}}

\newcommand{\CP}{\text{\sc C}[\WW]}

\newcommand{\VS}{\mathcal{V}}

\newcommand{\enc}[1]{\text{\sc enc}(#1)}
\newcommand{\lin}[1]{\text{\sc Lin}(#1)}
\newcommand{\otu}[2]{#1^{#2}}
\newcommand{\pfin}[1]{\mathcal{P}_\text{\tiny fin}(#1)}
\newcommand{\degpol}[2]{\text{\sc deg}_{#1}\, #2}
\newcommand{\degmon}[1]{\text{\sc deg}\, #1}
\newcommand{\wqo}{\text{\sc wqo}\xspace}
\newcommand{\age}[1]{\text{\sc Age}(#1)}
\newcommand{\agepar}[2]{\text{\sc Age}(#1, #2)}
\newcommand{\id}[1]{\textrm{Id}_{#1}}
\newcommand{\eqidealgen}[1]{\langle #1 \rangle}
\newcommand{\size}[1]{|#1|}
\newcommand{\prof}[3]{(#2,#3)\textsc{-prof}(#1)}
\newcommand{\restr}[2]{{#1}{|#2}}
\renewcommand{\a}{\alpha}
\renewcommand{\c}{c}
\renewcommand{\d}{c'}

\newcommand{\Z}{\mathbb{Z}}

\renewcommand{\b}{\beta}

\newcommand{\N}{\mathbb{N}}
\newcommand{\Qf}{\mathbb{Q}}
\newcommand{\Q}{Q}
\newcommand{\QQ}{\mathcal{Q}}

\newcommand{\w}{a}
\renewcommand{\v}{b}
\newcommand{\ww}{c}
\newcommand{\game}[2]{\mathcal{G}(#1,#2)}

\newcommand{\Spoiler}{\text{\sc Spoiler}\xspace}
\newcommand{\Reducer}{\text{\sc Reducer}\xspace}
\newcommand{\gamered}{\ll}
\newcommand{\B}{\mathcal{B}}
\newcommand{\Sof}[1]{\mathcal{S}(#1)}
\newcommand{\Ssub}[1]{\mathcal{S}_{#1}}
\newcommand{\Pres}[1]{\mathcal{P}_{#1}}

\newcommand{\I}{\mathcal{I}}
\newcommand{\J}{\mathcal{J}}
\newcommand{\R}{\mathbb{K}}
\newcommand{\Rf}{\mathbb{R}}
\newcommand{\WW}{\mathcal{A}}

\newcommand{\VV}{\mathcal{B}}

\newcommand{\W}{A}
\newcommand{\Vs}{B}
\newcommand{\Bs}{B}

\newcommand{\Us}{C}

\newcommand{\E}{{\embed{\WW}}}
\newcommand{\M}{\mathcal{M}}

\newcommand{\embed}[1]{\text{\sc Emb}_{#1}}

\newcommand{\Gr}{Gr{\"o}bner\xspace}
\newcommand{\Fr}{Fra{\.i}ss{\'e}\xspace}
\newcommand{\Bu}{Buchberger\xspace}

\newcommand{\vars}[1]{\text{\sc vars}(#1)}
\newcommand{\red}[3]{#1 \longrightarrow_#2 #3} 
\newcommand{\redstar}[3]{#1 \longrightarrow^*_#2 #3}
\newcommand{\finitedomain}[1]{\langle#1\rangle}

\newcommand{\lcm}[2]{\text{\sc lcm}(#1,#2)}
\newcommand{\lc}[1]{\text{\sc lc}(#1)}
\newcommand{\lm}[1]{\text{\sc lm}(#1)}
\newcommand{\cc}[1]{\text{\sc cp}(#1)}
\newcommand{\ltsymb}{\text{\sc lt}}
\newcommand{\lt}[1]{\ltsymb(#1)}

\newcommand{\polyring}[2]{#1[#2]}
\newcommand{\monomials}[1]{\text{\sc Mon}[#1]}

\newcommand{\set}[1]{\left\{#1\right\}}

\newcommand{\setof}[2]{\left\{#1 \, \middle\vert \, #2\right\}}
\newcommand{\zeromon}{\mathbf{1}}
\newcommand{\zeropol}{\mathbf{0}}

\newcommand{\subseteqfin}{\subseteq_\text{fin}}

\newcommand{\setfromto}[2]{\set{#1\ldots #2}}
\newcommand{\setto}[1]{\setfromto 0 {#1}}
\newcommand{\domprod}{\otimes}
\newcommand{\domexp}[2]{{#1}^{(#2)}}
\newcommand{\lexorder}{\wo_\text{\tiny lex}}

\newcommand{\lexorderstrict}{\wostrict_\text{\tiny lex}}
\newcommand{\lexorderstrictinv}{\wostrictinv_\text{\tiny lex}}
\newcommand{\lexorderinv}{\woinv_\text{\tiny lex}}
\newcommand{\divorder}{\mid}

\newcommand{\wpo}{\sqsubseteq}

\newcommand{\wo}{\preceq}

\newcommand{\wostrict}{\prec}
\newcommand{\woinv}{\succeq}
\newcommand{\wostrictinv}{\succ}

\newcommand{\detorder}{\rightsquigarrow}
\newcommand{\embset}[1]{\restr \E #1}
\newcommand{\embsetprod}[2]{#1\boxtimes #2}
\newcommand{\prettyexists}[2]{\exists #1 : #2}

\newcommand{\defeq}{\stackrel{\text{\tiny def}}{=}}

\newtheoremrep{theorem}{Theorem}
\newtheorem{property}[theorem]{Property}
\newtheorem{definition}[theorem]{Definition}

\newtheoremrep{lemma}[theorem]{Lemma}

\newtheoremrep{claim}[theorem]{Claim}
\newtheorem{corollary}[theorem]{Corollary}

\newtheorem{example}[theorem]{Example}
\newtheorem{remark}[theorem]{Remark}

\newenvironment{claimproofnew}{\paragraph{\textnormal{Proof:}}}{
\medskip
\hfill $\triangleleft$ 
}
\newcommand{\para}[1]{\paragraph{\rm \bf #1.}}

\newcommand{\probname}[1]{{\sc #1}}

\begin{document}


\title{Equivariant ideals of polynomials}

%
\author{Arka Ghosh}
\authornote{Partially supported by NCN grants 2019/35/B/ST6/02322 and 2022/45/N/ST6/03242.}
\orcid{0000-0003-3839-8459}
\affiliation{%
  \institution{University of Warsaw}
  \country{Poland}}
\author{S{\l}awomir Lasota}
\authornote{Partially supported by NCN grant 2021/41/B/ST6/00535 and by ERC
Starting grant INFSYS, agreement no. 950398.}
\orcid{0000-0001-8674-4470}
\affiliation{%
  \institution{University of Warsaw}
  \country{Poland}}



\begin{abstract}
We study existence and computability of finite bases for ideals of polynomials over infinitely many variables. In our setting, variables come from a countable logical structure A, and embeddings from A to A act on polynomials by renaming variables. First, we give a sufficient and necessary condition for A to guarantee the following generalisation of Hilbert's Basis Theorem: every polynomial ideal which is equivariant, i.e. invariant under renaming of variables, is finitely generated. Second, we develop an extension of classical Buchberger's algorithm to compute a Gr\"obner basis of a given equivariant ideal. This implies decidability of the membership problem for equivariant ideals. Finally, we sketch upon various applications of these results to register automata, Petri nets with data, orbit-finitely generated vector spaces, and orbit-finite systems of linear equations.
\end{abstract}


%



\maketitle


\section{Introduction}

A ring is \emph{Noetherian} if each of its ideals
is finitely-generated.
According to Hilbert's Basis Theorem,
the ring of polynomials over a finite set of variables,
with coefficients from a given Noetherian ring, is Noetherian.
In particular, every polynomial ideal can be represented by its finite basis, and hence can be given as input to algorithms.
Given such a finite basis, \Bu's algorithm computes a \Gr basis that can be used to decide
the
membership problem for polynomial ideals: determine, whether a given polynomial belongs to the
ideal generated by a given set of polynomials.
Hilbert's Basis Theorem and \Gr bases have important algorithmic applications, for instance
to decide zeroness of polynomial automata (weighted automata with
polynomial updates) \cite{Ben17},
reachability of reversible Petri nets \cite{MayrMeyer82}, or
bisimulation equivalence checking \cite{Stribrna97}.
Other applications of \Bu's algorithm are overviewed in \cite{Bu-app98}.

We generalise the above classical results to the ring of polynomials 
over an infinite set of variables.
Clearly,
Hilbert's Basis Theorem is not generalisable literally, since every polynomial belonging to 
a finitely-generated ideal may use at least one of variables used by the generators.
In this paper we investigate the following setting:
\begin{itemize}
\item
Variables are elements of a countable relational structure $\WW$
(which we call \emph{domain of variables}).
\item
The monoid of embeddings $\WW \to \WW$ acts on polynomials by renaming variables.
\item 
We restrict to polynomial ideals which are invariant under renaming (such ideals we call
\emph{equivariant}).
Accordingly, we adapt the definition of basis, enforcing equivariance of a generated ideal.
\end{itemize}

\begin{slexample}  \label{ex:pol1}
For illustration, consider \emph{equality} domain of variables, namely
the structure $\WW = (\W, =)$ consisting of a countable infinite set 
$\W$ 
and the equality relation.
Embeddings are injective mappings $\W\to\W$.
Consider the ring $\polyring \Rf \WW$ of polynomials with real coefficients over variables $\W$.
The subset $\I\subseteq\polyring \Rf \WW$ of polynomials whose coefficients sum up to 0, is an equivariant ideal.
It is the least ideal containing 
all univariate polynomials
$p_\w(\w) = \w - 1$, where $\w\in\W$.
For instance, the polynomial $p(a,b,c) = ab^2 - bc$ with variables
$a,b,c\in\W$ belongs to $\I$, and decomposes as:
\[
ab^2 - bc \ = \  
(b-1)(ab + a - c) + (a-1) - (c-1).
\]
$\I$ is also the least \emph{equivariant} ideal including  $\set{p_\w}$, 
for a fixed $\w\in\W$,
hence finitely-generated according to our adapted definition.
\end{slexample}

\para{Contribution}
We investigate equivariant ideals of polynomials over domains of variables
which are \emph{totally ordered}, namely admit a total order as one of its relations, and
\emph{well structured}, namely admit a naturally defined
\emph{well quasi order} (as defined in Section \ref{sec:prelim}).
A prototypical example of such a domain is rational order $\QQ=(\Q, \leq_\QQ)$.
For such domains $\WW$ we prove, in Section \ref{sec:equiv noeth}, 
the following generalisation of Hilbert's Basis Theorem:
for every Noetherian ring $\R$, every equivariant ideal in
the polynomial ring $\polyring \R \WW$ 
has a finite basis.
We transfer this property to lexicographic products of totally ordered, well structured domains,
and to reducts thereof,
including the equality domain. 
We also show that well structure is necessary for the generalisation of Hilbert's Basis Theorem to hold.

For further investigations we restrict to domains of variables which are
\emph{well ordered}, well structured, and  
\emph{computable} (as formalized in Section \ref{sec:comp}).
We develop a generalisation of \Bu's algorithm to compute a \Gr basis of an
equivariant ideal, given as a finite set of generators.
As a corollary, in Section \ref{sec:decid} we obtain
decidability of the membership problem for equivariant ideals.
Similarly as our generalisation of Hilbert's Basis Theorem, we transfer
the decidability result to a wide range of other domains, 
including equality domain, rational order, and lexicographic products thereof.

Our results are of fundamental nature, with various potential applications.
We sketch upon some of them.
In Section \ref{sec:lin} we apply our generalisation of Hilbert's Basis Theorem to prove
that orbit-finitely generated vector spaces have finite length, i.e., admit no infinite increasing chains
of equivariant subspaces.
This generalises a crucial result of \cite{BKM21}
(however, without explicit upper bound).
As a consequence, we derive decidability of
the zeroness problem for weighted register automata, again 
generalising a result of \cite{BKM21}.
We also apply, in Section \ref{sec:pn}, the ideal membership problem 
to decide reachability in reversible Petri nets with data, for a wide range of data domains.
This includes rational order domain, where the problem is known to be undecidable 
for general non-reversible Petri nets \cite{LNORW07}.
Finally, 
in Section \ref{sec:eq}, we apply the ideal membership problem 
to solving orbit-finite systems of linear equations.
This generalises a core result of \cite{GHL22} to a wide range of domains
(again, without complexity upper bound).
All our proofs of generalisations of \cite{BKM21} and \cite{GHL22}
are remarkably simple, compared to the original ones.

\para{Structure of the paper}
The paper has two parts: the first part concentrates on \emph{existence} of finite bases,
while the second one investigates \emph{computability} thereof.
In the first part, Sections \ref{sec:prelim}--\ref{sec:lin},
we set up preliminaries, prove our generalisation of Hilbert's Basis Theorem,
and sketch upon its applications.
Then, in Sections \ref{sec:comp}--\ref{sec:eq}
we generalise \Bu algorithm, use it for decidabillity of the ideal membership problem, 
and show applications thereof.
Missing proofs are to be found in the full version of this paper \cite{GL24-arxiv}.

\para{Related research}
Our motivations, together with the most fundamental concepts such as equivariance and orbit-finiteness, 
come from the theory of sets with atoms \cite{BKL14,atombook}.
The concept of well structured domain, crucial for development of this paper,
turns out to be essentially the same as \emph{preservation of well quasi order}, 
introduced in \cite{Lasota16} and conjectured to characterise data domains
rendering decidability of a range of decision problems for Petri nets.
Relationship between the two settings calls for further investigation.

There have been 
numerous earlier contributions regarding existence and computability of finite bases of equivariant ideals.
To the best of our knowledge, all of them considered 
(using our terminology) either equality or well order domains.
First, \cite{Cohen67} proves finite basis property for domain $\omega = (\N, \leq)$, then
\cite{Cohen87} and \cite{Emmott87} extend to larger ordinals.
These two papers also give sufficient condition for computability of \Gr bases,
which is shown to hold in $\omega$ in \cite{HKL18}.
Similar results were independently obtained in \cite{AH07,AH08,BD11}.
In \cite{HS12} the authors use the aforementioned results to prove the Independent Set Conjecture in algebraic statistics.

There are a number of recent results concerning finite generation of vector spaces:
\cite{BKM21} shows that every equivariant subspace of orbit-finitely generated vector space is finitely generated
(this follows from our results in Section \ref{sec:equiv noeth}, as argued in Section \ref{sec:lin});
furthermore, \cite{GHL22} shows that dual vector spaces are finitely generated.

%

\section{Preliminaries}  \label{sec:prelim}


A relational structure $\WW = (\W, r_1, \ldots, r_n)$ 
is a set $\W$ together with finitely many relations 
$r_1\subseteq \W^{s_1}, \ldots, r_n\subseteq \W^{s_n}$ of finite arities $s_1, \ldots, s_n\in\N$, 
respectively.
In the sequel we restrict to countable structures only.
Every relation $r\subseteq \W^m$ which is expressible as a Boolean combination of relations
$r_1, \ldots, r_n$, we call \emph{definable}.
We always assume that equality $=$ is definable.
An \emph{induced substructure} of $\WW$ is the restriction of $\WW$ to some subset of $\W$.

An \emph{embedding} of $\WW$ is a map $\iota:\W\to\W$ that preserves and reflects all relations,
which means that for every relation $r_i$ and
$a_1, \ldots, a_{s_i}\in\W$, the following equivalence holds:
\[
r_i(a_1, \ldots, a_{s_i}) \iff r_i(\iota(a_1), \ldots, \iota(a_{s_i})).
\]
In other words, an embedding $\iota$ is an isomorphism of $\WW$ and its substructure 
induced by $\iota(\W)$.
Clearly, embeddings preserve and reflect all definable relations, and hence each such relation
can be considered w.l.o.g.~ to be among $r_1, \ldots, r_n$.
The set of all embeddings of $\WW$ we denote as $\embed \WW$.
We say that $\WW$ is \emph{totally ordered} (resp.~\emph{well ordered})
if a total order 
(resp.~a well order) 
is definable in  $\WW$.

Speaking intuitively, a \emph{reduct} of $\WW$ is a structure obtained by removing 
some relations from $\WW$. 
Formally,
a structure $\WW' = (\W, r_1, \ldots, r_n)$ is a reduct of $\WW = (\W, \ldots)$, with the same carrier 
set $\W$,
if all relations $r_1, \ldots, r_n$ are definable in $\WW$.


\begin{slexample} 
For a countable infinite set $\W$, embeddings of the equality structure $\WW=(\W, =)$, 
with equality as the only relation, are
all injective maps $\W\to\W$.
Embeddings of the rational order  $\QQ = (\Q, \leq_\QQ)$ are all strictly monotonic 
maps $\iota:\Q\to\Q$, i.e.~functions satisfying 
\[
q<_\QQ q'\implies \iota(q)<_\QQ \iota(q').
\] 
Note that equality is definable using $\leq$, namely $q = q' \iff q\leq q' \wedge q' \leq q$.
The equality structure is thus a reduct of  $\QQ$.

Every ordinal $\a$ can be seen as a (well ordered) relational structure, with just one binary relation --
the ordinal order.
Again, its embeddings are strictly monotonic maps $\a\to\a$.
A special case is $\omega = (\N, \leq) = \set{0 < 1 < \ldots}$;
or $\finitedomain d = \set {0 < 1 < \ldots < d{-}1}$, for $d\in\N$, where the only embedding
is the identity.
In the sequel we identify an ordinal with the corresponding relational structure.
\end{slexample}


\para{Noetherianity}

Recall that an \emph{ideal} in a commutative ring $\R$ is any nonempty subset $\I\subseteq\R$ 
closed under addition and multiplication by any element of $\R$: $\w, \w' \in \I {\implies} \w+\w'\in\I$,
and $\w\in\I, \v\in\R{\implies}\w\cdot\v\in\I$.
A basis of $\I$ is any subset $\B\subseteq\I$ such that every $\w\in\I$ is obtainable from $\B$
by addition and multiplication by arbitrary elements of $\R$;
equivalently, $\w= \w_1 \cdot \v_1 + \ldots + \w_n \cdot \v_n$ for some $\v_1,\ldots, \v_n\in\B$ and
$\w_1, \ldots, \w_n\in\R$.
A commutative ring is \emph{Noetherian} if each of its ideals has a finite basis.
%
In particular, each field is a Noetherian commutative ring.

\para{Polynomials}

Let $\WW = (\W, r_1,\ldots, r_n)$ be a relational structure,
which we call a \emph{domain of variables}, or simply a \emph{domain}.
We define, in the expected way,
the commutative multiplicative monoid of monomials over variables from $\W$, 
denoted $\monomials \WW$. 
A \emph{monomial} is a formal expression of the form
\begin{align*} 
f(\w_1 \ldots \w_m) = \w_1^{n_1} \cdot  \ \ldots \  \cdot  \w_m^{n_m},
\end{align*}
where $m\geq 0$, $\vars f = \set{\w_1, \ldots, \w_m} \subseteq \W$ are distinct variables, 
and $n_i = \degpol f {\w_i}>0$ is the \emph{degree} of variable $\w_i$ in $f$.
For $\w \notin\vars f$ we set $\degpol f \w = 0$.
The trivial unit monomial, for $m=0$, we denote as $\zeromon$.
Also, given a commutative ring $\R$,
we define the commutative ring of polynomials with coefficients from $\R$ and variables from $\W$, 
denoted $\polyring \R \WW$.
A \emph{term} is a formal expression of the form $r\cdot f$, for $r\in\R$ and $f\in\monomials \WW$,
and a \emph{polynomial} is a finite sum of terms, i.e., a formal expression of the form
\[
f \ = \ r_1 \cdot f_1 + \ldots + r_n \cdot f_n,
\]
$f_1, \ldots, f_n\in\monomials \WW$ are distinct monomials and
$r_1, \ldots, r_n\in\R$ are their \emph{coefficients}.
The zero polynomial, for $n=0$, we denote by $\zeropol$.
We set $\vars f = \vars {f_1} \cup \ldots \cup \vars {f_n}$.

\para{Equivariant Noetherianity}

Embeddings act on polynomials by renaming variables: 
an embedding $\iota \in \embed {\W}$ maps a polynomial $f$ to the polynomial $\iota(f)$ 
obtained by renaming all
variables in $f$ according to $\iota$, i.e., by replacing each variable $\w$ by $\iota(\w)$%
\footnote{
Clearly, $f$ and $\iota(f)$ define the same functions $\R^n\to\R$, where $n=\size{\vars f}$.
}.
%
An ideal $\I\subseteq \polyring \R \WW$ is called \emph{equivariant} if it is additionally closed under
the action of $\E$:
whenever $f\in \I$ and $\iota \in \E$ then  $\iota(f) \in \I$.
We adapt the notion of basis by taking additionally into account the action of $\E$.
By a \emph{basis} of an equivariant ideal $\I$ we mean a subset 
$\B\subseteq \I$ such that each element of $\I$ is obtainable from $\B$ by addition, multiplication 
by arbitrary polynomials, and action of $\E$.
This is equivalent to saying that
\begin{align} \label{eq:basis}
\E(\B) \ \defeq \ \setof{\iota(g)}{\iota \in \E, \ g\in\B}
\end{align}
is a basis of $\I$ in classical sense, or to presentability of 
every element $f\in\I$ as
\[
f \ = \ \sum_{i=1}^{n} h_i \cdot \iota_i(g_i),
\]
for $g_1 \ldots g_n \in \B$ and $h_1 \ldots h_n \in \polyring \R \WW$ and $\iota_1 \ldots \iota_n \in \E$
(since the action of $\E$ commutes with addition and multiplication of polynomials).
When $\B$ is a basis of $\I$, 
we also say that $\I$ is the equivariant ideal \emph{generated} by $\B$, and denote it by
$\I = \eqidealgen \B$.

\begin{slexample} \label{ex:Q}
Consider the rational order $\QQ = (\Q, \leq_\QQ)$ as a domain of variables, 
and a polynomial
\[
f_0(\w_0, \v_0, \ww_0) \ = \  \v_0^2 - \w_0 \ww_0
\]
for some arbitrarily chosen, fixed variables $\w_0 < \v_0 < \ww_0$ from $\Q$.
For every other variables $\w < \v < \ww$, the polynomial
$
\w \ww^2 - \v^3
$
belongs to the equivariant ideal generated by $\set{f_0}$, because
\[
\w \ww^2 - \v^3
\ = \ 
\w \cdot (\ww^2 - \v\vv) \ - \ 
\v \cdot  (\v^2 - \w \vv),
\]
for arbitrary variable $\vv > \ww$.
\end{slexample}

The ring $\polyring \R \WW$ we call \emph{equivariantly Noetherian} if each equivariant ideal has a finite basis.
In Section \ref{sec:equiv noeth} we demonstrate that a wide range of domains $\WW$ satisfy
the following generalisation of Hilbert's Basis Theorem:

\begin{property}
\label{prop:noeth}
For every Noetherian commutative ring $\R$, the polynomial ring $\polyring \R \WW$ is equivariantly Noetherian.
\end{property}

\para{Orbit-finite bases}

Equivalently, $\polyring \R \WW$ is equivariantly Noetherian if each equivariant ideal 
has a basis, in classical sense, which is \emph{orbit-finite}, i.e., finite up to embeddings.
Formally, orbit-finite sets are set of the form \eqref{eq:basis} where $\B$ is finite.
A finite set $\B$ is called a \emph{presentation} of an orbit-finite set \eqref{eq:basis}.
The concept of orbit-finiteness comes from the theory of sets with atoms \cite{atombook,BKL14},
where it refers to sets which are infinite but finite up to automorphisms of an underlying structure
of atoms.
We slightly extend this setting here, and consider embeddings instead of just automorphisms.
In case of equality and rational order domains (and any \emph{homogeneous} structure),
the action of embeddings is essentially the same as the action of automorphisms, 
as every automorphism of finite induced substructures extends to an automorphism of the whole domain.

\para{Ideal membership}

For $\WW$ satisfying Property \ref{prop:noeth}, every
equivariant ideal in $\polyring \R \WW$ can be represented by 
its finite basis, and thus can be input to algorithms. 
In Section \ref{sec:decid}, assuming $\R$ to be a fixed computable field (e.g.~rational%
\footnote{
The set of rationals appears in the sequel in two roles: as a totally ordered set
$\QQ$, and 
as a field $\Qf$. 
We use different fonts
to distinguish these two roles.
}
field $\Qf$),
we study decidability of the following decision problem
of membership of a given polynomial in a given equivariant ideal, 
for different domains $\WW$:

\begin{description}                                                                                                                                                               
\item[\probname{Ideal-memb in $\polyring \R \WW$}:]           \    
\item[\bf Input: ] \ \ \ \ \ \ \ a polynomial $f\in\polyring \R \WW$ and an equivariant ideal 
\item[] \ \ \ \ \ \ \ \ \ \ \ \ \ \ \ \ $\I\subseteq \polyring \R \WW$.
\item[\bf Question: ] $f\in\I$?
\end{description}

\para{Well structured structures} 

A quasi order consists of a set $Y$ and a reflexive and transitive
relation $\preceq\,\subseteq Y\times Y$.
$(Y, \preceq)$ is a \emph{well quasi order} (\wqo) 
if 
every infinite sequence $y_1, y_2, \ldots$ of elements of $Y$ contains a dominating pair,
i.e., two elements $y_i$ and $y_j$, for $i<j$, such that $y_i \preceq y_j$.
Clearly, every well order is a \wqo.

A subset $X\subseteq Y$ is \emph{upward-closed} if $x\in X$ and $x\preceq x'$ implies $x'\in X$.
A \emph{basis} of an upward-closed set $X$ 
(not to be confused with a basis of an ideal) is any subset $B\subseteq X$ such that
for every $x\in X$ there is $b\in B$ with $b\preceq x$.
Every \wqo is well-founded and admits no infinite antichains
\cite[Proposition 12.1.1]{diestel}, which implies:

\begin{lemma} \label{lem:wqo}
Every upward-closed subset of a \wqo has a finite basis.
\end{lemma}

Given a relational structure $\WW$ and a quasi order $(Y, \preceq)$, 
by $\age \WW$ we denote the set of all its finite induced substructures.
We do not distinguish between a finite subset $\Vs\subseteqfin \W$ and the substructure 
$\VV$ of $\WW$ induced by $\Vs$, i.e., we equate $\age \WW$ with $\pfin \WW$.
By $\agepar \WW Y$ we denote the set of all finite induced substructures labeled by $Y$,
i.e., functions $\ell : \Vs \to Y$ where $\Vs  \in \age \WW {}$.
The quasi order $(Y, \preceq)$ \emph{determines} a natural quasi order $\detorder$ on
$\agepar \WW Y$, where $(\ell : \Vs \to Y) \detorder (\ell' : \Vs' \to Y)$ if there is an embedding 
$\iota\in\E$
such that $\iota(B)\subseteq B'$ and
$\ell(\v) \preceq \ell'(\iota(\v))$ for all $\v \in \Vs$.
In most cases we will focus on $(Y, \preceq) = (\N, \leq)$.

\begin{sldefinition}
$\WW$ is  \emph{$\omega$-well structured} if
$\omega = (\N, \leq)$ determines a \wqo on $\agepar \WW \N$.
Moreover, $\WW$ is \emph{well structured} if
every \wqo $(Y, \preceq)$ determines a \wqo on $\agepar \WW Y$.%
\footnote{
It is possible that the two conditions are equivalent.
See e.g.~\cite{KS91} for related results.
}
\end{sldefinition}

For instance, the rational order  $\QQ = (\Q, \leq_\QQ)$ is well structured. 
Indeed,
finite induced substructures of $\QQ$ are finite total orders, and $\agepar \QQ Y$ is a \wqo
by Higman's Lemma \cite{Hig52}.
The equality structure is also well structured, as a reduct of $\QQ$.
\emph{Universal tree}
arising as  \Fr limit of finite trees \cite[Section 7]{atombook},
similarly as $\QQ$ arises as \Fr limit of finite total orders,
is well structured.
%
By Nash-Williams's Theorem \cite[Theorem 2]{nash-williams} 
(cf.~also \cite{Hig52} in the restricted case of $\omega$)
ordinals are well structured too: 
\begin{lemmarep} \label{lem:wpo}
Every ordinal (seen as a relational structure) is well structured.
\end{lemmarep}
\begin{proof}
Let $\a$ be a fixed ordinal and $(Y,\preceq)$ a fixed \wqo. 
We add to $Y$ the least element $*$, namely put $*\prec y$ for all $y\in Y$, and observe
that $(Y\cup \set{*}, \preceq)$ is still a \wqo.
According to Theorem 2 in \cite{nash-williams}, the \wqo
$(Y\cup\set{*},\preceq)$ determines a \wqo on the set $F(Y)$ of image-finite functions 
$\ell:\a\to Y\cup\set{*}$.
Elements of $\agepar \a Y$  may be identified with those functions $\ell:\a\to Y$
which map almost all elements to 0, a subset $F_0(Y) \subseteq F(Y)$.
Furthermore, the quasi orders determined by $(Y\cup\set{*},\preceq)$ on $\agepar \a Y$ and 
$F_0(Y)$ agree.
Therefore $(Y\cup\set{*},\preceq)$ determines a \wqo on $\agepar \a Y$, as required.
\end{proof}

\para{Lexicographic product of structures}  \label{sec:produ}

We will build larger domains from smaller ones, 
using  \emph{lexicographic product} defined below.
The operation, in contrast to Cartesian product
 (as illustrated in Example \ref{ex:Cartesian}), preserves well structure.

Given two relational structures $\WW = (\W, \ldots)$ and $\WW' = (\W', \ldots)$, 
their lexicographic product
$\WW \domprod \WW'$ is defined
as follows.
The carrier set of $\WW \domprod \WW'$ is the Cartesian product $\W\times\W'$,
and the set of relation symbols is the disjoint union of sets of relation symbols of $\WW$ and $\WW'$.
For each $k$-ary relation $r$ of $\WW$, the relation $r$ in $\WW \domprod \WW'$ 
relates those $k$-tuples
\begin{align}\label{eq:tuples}
((\w_1, \w'_1), \ldots, (\w_k, \w'_k))
\end{align}
where $r(\w_1, \ldots, \w_k)$ in $\WW$ (the second coordinate comming from $\WW'$ is ignored). 
Furthermore, 
for each $k$-ary relation $r'$ of $\WW'$, the relation $r'$ in $\WW \domprod \WW'$ relates those tuples
\eqref{eq:tuples}
where $\w_1 = \ldots = \w_k$ and $r'(\w'_1, \ldots, \w'_k)$ in $\WW'$
(the first coordinate comming from $\WW$ is required to be equal).
The product is associative.

\begin{lemmarep}\label{lem:prod embed}
Embeddings of $\WW \domprod \WW'$ are exactly maps of the form
\begin{align} \label{eq:prod embed}
(\w, \w') \ \longmapsto \ (\iota(\w), \iota'_{\w}(\w')),
\end{align}
determined by some $\iota\in\E$ and
a family $\setof{\iota'_{\w}\in\embed{\WW'}}{\w\in\W}$.

\end{lemmarep}
\begin{proof}
In one direction, every mapping $\W\times\W' \to \W\times\W'$ of the form \eqref{eq:prod embed}
preserves and reflects all relations of $\WW\domprod\WW'$, since each relation of $\WW$ is
preserved and reflected by $\iota\in\embed \WW$, and each relation of $\WW'$ is preserved
and reflected by every $\iota'_\w$, for $\w\in\W$.

In the reverse direction, consider any embedding $\kappa \in \embed{\W\domprod\W'}$.
As equality is definable in $\WW$, the mapping $\kappa$ preserves and reflects equality of the first coordinate, which determines an injection $\iota: \W\to\W$
such that the first component of $\kappa(\w,\w')$ equals $\iota(\w)$.
As $\kappa$ preserves and reflects all relations of $\WW$, the injection is an embedding,
$\iota\in\E$.
Furthermore, for each $\w\in\W$ the mapping $\kappa$ restricted to pairs $(\w, \_)$,
preserves and reflects all relations of $\WW'$ on the second coordinate.
Therefore, for each $\w\in\WW$ there is some $\iota'_\w\in\embed{\WW'}$ such that $\kappa$ is equal to
\eqref{eq:prod embed}.
\end{proof}

\begin{lemmarep} \label{lem:prod}
Lexicographic product of two totally ordered (well ordered, well structured, resp.) relational structures is 
totally ordered (well ordered, well structured, resp.).
\end{lemmarep}
\begin{proof}
Consider two relational structures $\WW = (\W, \ldots)$ and $\WW' = (\W', \ldots)$.
Let $\wo$ and $\wo'$ be definable total orders on $\WW$ and $\WW'$, respectively.
The lexicographic product of $\wo$ and $\wo'$ is definable in $\WW \domprod \WW'$,
and is well when $\wo$ and $\wo'$ are well, and therefore
if both $\WW$ and $\WW'$ are totally ordered (resp.~well ordered), so is $\WW \domprod \WW'$.

It thus remains to show that $\WW \domprod \WW'$ is well structured.
This follows immediately, once we prove:

\begin{claim}  \label{claim:wqo emb}
Let $(Y, \prec)$ is a \wqo.
Then $\agepar {\WW \domprod \WW'} Y$ is isomorphic, as a quasi order, 
to $\agepar \WW {\agepar {\WW'} Y}$.
\end{claim}
%
For proving the claim,
consider any $\Vs \in \age {\WW \domprod \WW'}$
and its labelling $(\ell : \Vs \to Y) \in \agepar{\WW\domprod\WW'}{Y}$.
Let $\overline\Vs \in\age \WW$ be its projection 
on the first coordinate.
Furthermore, for $\w\in\overline\Vs$ let 
$\Vs_\w \defeq \setof{\w'}{(\w,\w')\in\Vs}\in \age {\WW'}$, and let $\ell_\w : \Vs_\w \to Y$ be the
corresponding restriction of $\ell$, namely $\ell_\w(\w') = \ell(\w,\w')$.
This defines a function
\[
\overline\ell : \w \mapsto \ell_\w, \qquad
\text{ for } \w \in \overline\Vs,
\]
i.e.,
$(\overline \ell : \overline\Vs \to \agepar {\WW'} Y) \in \agepar \WW {\agepar {\WW'} Y}$.
We have thus defined a mapping $\ell \mapsto \overline\ell$.
One easily defines its inverse, and hence it
is a bijection between
$\agepar{\WW\domprod\WW'}{Y}$ and $\agepar \WW {\agepar {\WW'} Y}$.

Given two induced substructures
$\Vs_i \in \age {\WW \domprod \WW'}$
and its labellings
$(\ell_i : \Vs_i \to Y) \in \agepar{\WW\domprod\WW'}{Y}$, for $i = 1, 2$,
using Lemma \ref{lem:prod embed} we notice equivalence of the following two conditions:
\begin{itemize}
\item
$\ell_1 \detorder \ell_2$ are related by the quasi order determined  by $(Y,\preceq)$
on $\agepar{\WW\domprod\WW'}{Y}$;
\item
there is some $\iota\in\E$ such that
for every $\w\in\overline\Vs$, $(\ell_1)_\w \detorder (\ell_2)_{\iota(\w)}$ are related
by the quasi order determined by $(Y, \preceq)$ on $\agepar {\WW'} Y$.
\end{itemize}
The latter condition says exactly that $\overline\ell_1 \detorder \overline\ell_2$ are related
by the quasi order determined by $(Y, \preceq)$ on $\agepar \WW {\agepar {\WW'} Y}$.
The claim is thus shown.
\end{proof}

\begin{slexample} \label{ex:Cartesian}
Since the rational order  $\QQ$ is well structured,  
the lexicographic product
$\QQ\domprod \QQ$ is also so, while
the Cartesian product $\QQ\times\QQ$ is not.
We claim that the quasi order $\detorder$
determined by an arbitrary \wqo on $\age{\QQ\times\QQ}$ (the choice of \wqo is irrelevant)
is not a \wqo.
Indeed,
consider the family of 'cycle' subsets $C=(C_n\in\age{\QQ\times\QQ})_{n\in\N}$, where
(cf.~Figure \ref{fig:Cn})
\begin{align*}
C_n = \setof{(x_i, y_i), (x_i, y_{i+1\!\!\!\mod n})}{i\in\setto {n-1}}
\end{align*}
for arbitrary rationals $x_0< \ldots < x_{n-1}$ and $y_0 < \ldots < y_{n-1}$.
As embeddings of $\QQ\times\QQ$ are induced by pairs of embeddings of $\QQ$,
\[
\embed{\QQ\times\QQ} = \embed \QQ \times \embed \QQ,
\]
a smaller cycle can not embed into a larger one.
Therefore $C_n \detorder C_m$ implies $n=m$, i.e., $C$ is an antichain,
and hence $\detorder$ is not a \wqo.
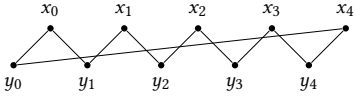
\begin{figure}
\scalebox{.85}{
\begin{tikzpicture}[every node/.style={circle, inner sep=1pt, node distance=20pt, fill=black}]

\newcommand*{\n}{5} 
\pgfmathtruncatemacro\nm{\n-1}
\pgfmathtruncatemacro\nmm{\nm-1}
\node (y0) [label=below :{$y_0$}] {};
\foreach \i in {0,...,\nmm}
{
  \node (x\i) [above right=of y\i, label=above:{$x_{\i}$}] {};
  \pgfmathtruncatemacro\ij{\i+1}
  \node (y\ij) [below right=of x\i, label=below:{$y_{\ij}$}] {};
}
\node (x\nm) [above right=of y\nm, label=above:{$x_{\nm}$}] {};
\foreach \i in {0,...,\nm}
{
  \draw[-] (x\i)--(y\i);
  \pgfmathtruncatemacro\ij{\i+1}
  \pgfmathtruncatemacro\ij{Mod(\ij,\n)}
  \draw[-] (x\i) -- (y\ij);
}
\end{tikzpicture}
}
\caption{$C_5\in\age {\QQ\times\QQ}$. \rm Its elements are depicted by edges.}
\label{fig:Cn}
\end{figure}
\end{slexample}


\section{Equivariant Noetherianity} 
\label{sec:equiv noeth}

We prove that every totally ordered and $\omega$-well structured domain of variables satisfies  
Property \ref{prop:noeth}.
Moreover, almost conversely, Property \ref{prop:noeth} implies that domain of variables is 
$\omega$-well structured.

\begin{theorem} \label{thm:wqo}
If the polynomial ring $\polyring \R \WW$ is equivariantly Noetherian for some commutative ring $\R$ then
$\WW$ is $\omega$-well structured.
\end{theorem}

\begin{theoremrep}\label{thm:equiv noeth}
For every $\omega$-well structured and totally ordered domain of variables $\WW$ and a
Noetherian commutative ring $\R$, the polynomial ring $\polyring \R \WW$ is equivariantly Noetherian.
\end{theoremrep}
\begin{proof}
Suppose, towards contradiction, that an equivariant ideal $\I\subseteq\polyring\R\WW$ is not
finitely generated.
Then there exists an infinite sequence $f_0,f_1,f_2,\dots$ of polynomials in $\I$ such that
$f_n \in \I\setminus\eqidealgen{\set{f_0,\dots,f_{n-1}}}$ for all $n\in\N$.
Since $\lexorder$ is well-founded we can furthermore assume that $\cc{f_{n}}$ is $\lexorder$-minimal for all $n$,
i.e. there is no $g\in\I\setminus\eqidealgen{\set{f_1,\dots,f_{n-1}}}$ such that $\cc{g}\lexorderstrict \cc{f_n}$.
As $\wpo$ is a well partial order, the sequence $(\cc {f_n})_{n\in\N}$ contains an infinite non-decreasing subsequence:
\[
\cc{f_{n_1}} \ \wpo \ \cc{f_{n_2}} \ \wpo \ \dots
\]
Let $\J_k \subseteq \R$ be the ideal generated by $\{\lc{f_{n_1}},\dots,\lc{f_{n_k}}\}$.
We have thus an increasing chain of ideals in $\R$,
\[
\J_1 \ \subseteq \ \J_2 \ \subseteq \ \dots,
\] 
which necessarily stabilises as $\R$ is Noetherian.
Let $K \in \N$ be such that $\J_n = \J_K$ for all $n \geq K$.
As $\J_{K+1} = \J_K$,
there are $r_1,\dots,r_K \in \R$  such that
\begin{align} \label{eq:lclc}
\lc{f_{n_{K+1}}} = \sum_{i = 1}^K r_i \cdot \lc{f_{n_i}}.
\end{align}
Moreover, since $\cc {f_{n_i}} \wpo \cc {f_{n_{K+1}}}$ for $i = 1 \ldots K$, there are some
monomials $g_1,\dots,g_K \in \monomials \WW$ and 
embeddings $\iota_1,\dots,\iota_K \in \E$ such that
\begin{align*}
g_i \cdot \iota_i(\lm{f_{n_i}}) & = \lm{f_{n_{K+1}}} \\
\iota_i(\vars{f_{n_i}}) & \subseteq \vars{f_{n_{i+1}}} 
\end{align*}
for $i = 1 \ldots K$.
Using the fact that $\lm{\_}$ and $\vars{\_}$ both commute with the action of embeddings, 
and Lemma \ref{lem:ltcdot}, the left-hand side can be rewritten to:
\begin{align}\label{eq:vars and mon}
\lm{g_i \cdot \iota_i(f_{n_i})} & = \lm{f_{n_{K+1}}} \\
\label{eq:varsvars}
\vars{\iota_i(f_{n_i})} & \subseteq \vars{f_{n_{i+1}}} 
\end{align}
Since $\lt{f_{n_{i}}}  = 
\lc{f_{n_{i}}} \cdot \lm{f_{n_{i}}}$,
by the equalities \eqref{eq:lclc} and \eqref{eq:vars and mon}  we get:
\[
\lt{f_{n_{K+1}}} \ = \ 
\sum_{i = 1}^K r_i \cdot \lt{g_i \cdot \iota_i(f_{n_i})}.
\]
Using the equality $r \cdot \lt{f} = \lt{r \cdot f}$, 
we rewrite further to:
\begin{align*}
\lt{f_{n_{K+1}}} \ = \ 
\sum_{i = 1}^K \lt{ r_i \cdot g_i \cdot \iota_i(f_{n_i})}.
\end{align*}
The left-hand side of the above equality contains just one term,
therefore $\lm{f_{n_{K+1}}} = \lm{ r_i \cdot g_i \cdot \iota_i(f_{n_i})}$ for some $i\in\setfromto 1 K$,
and leading monomials different that $\lm{f_{n_{K+1}}}$ cancel out in the right-hand side sum.
Dropping summands with leading monomials different that $\lm{f_{n_{K+1}}}$ yields: 
%
\begin{align}\label{eq:lt}
\lt{f_{n_{K+1}}} \ = \ 
\ltsymb\big(\sum_{i = 1}^K r_i \cdot g_i \cdot \iota_i(f_{n_i})\big).
\end{align}
Let $f'$ be the polynomial appearing on the right-hand side:
\[
f' \ := \ \sum_{i = 1}^K r_i \cdot g_i \cdot \iota_i(f_{n_i}).
\]
As $f' \in \I_{n_K}$ and $f_{n_{K+1}} \in \I_{n_{K+1}} \setminus \I_{n_K}$,
the difference of polynomials
$(f_{n_{K+1}} - f') \in \I_{n_{K+1}} \setminus \I_{n_K}$.
All the leading terms appearing on both sides of
\eqref{eq:lt} cancel out in $f_{n_{K+1}} - f'$.
We thus deduce 
\[
\lm {f_{n_{K+1}} - f'} \lexorderstrict \lm{f_{n_{K+1}}}.
\]
Moreover, by \eqref{eq:vars and mon} and \eqref{eq:varsvars} we deduce
\[
\vars{f_{n_{K+1}} - f'} \subseteq \vars{f_{n_{K+1}}} \ .
\]
Combining the last two statements we obtain
\[
\cc {f_{n_{K+1}} - f'} \lexorderstrict \cc{f_{n_{K+1}}},
\]
a contradiction with $\lexorder$-minimality of $\cc{f_{n_{K+1}}}$, which completes the proof.
\end{proof}

\begin{lemma}\label{lem:equiv noeth}
If $\polyring \R \WW$ is equivariantly Noetherian and $\VV$ is a reduct of $\WW$ then
$\polyring \R \VV$ is also equivariantly Noetherian.
\end{lemma}
\begin{proof}
Since $\E \subseteq \embed \VV$, 
every equivariant ideal $\I$ in $\polyring \R \VV$ is automatically an equivariant ideal
in $\polyring \R \WW$, and therefore
$\I$  is finitely generated in there.
Furthermore, again due to $\E \subseteq \embed \VV$, a finite basis of $\I$ 
in $\polyring \R \WW$ 
is also its basis in $\polyring \R \VV$.
\end{proof}

\begin{remark} \label{rem:appliesto}
\rm
Theorem \ref{thm:equiv noeth}
applies in particular to: every ordinal (by Lemma \ref{lem:wpo}); rational order $\QQ$;
equality domain 
(by Lemma \ref{lem:equiv noeth});
lexicographic products thereof (by Lemma \ref{lem:prod}).
All subsequent results in this paper apply to all the above-listed domains too. 
\end{remark}

According to Pouzet's conjecture \cite{Pouzet20},  
necessary and sufficient conditions of Theorems \ref{thm:wqo} and \ref{thm:equiv noeth}
coincide, namely
 every $\omega$-well structured structure may be equipped with a total order
while preserving $\omega$-well structure:
\begin{conjecture}[\cite{Pouzet20}]
Every $\omega$-well structured relational structure $\WW$ is
a reduct of an $\omega$-well structured, totally ordered structure.
\end{conjecture}
The conjecture, together with Theorems \ref{thm:wqo} and \ref{thm:equiv noeth}, and 
Lemma \ref{lem:equiv noeth},
would imply that 
Property \ref{prop:noeth} is satisfied
exactly by 
reducts of $\omega$-well structured and totally ordered domains of variables.


\para{Quasi orders on monomials}

We need some definitions before proving Theorems \ref{thm:wqo} and \ref{thm:equiv noeth}.
In the sequel let $\R$ be a fixed Noetherian commutative ring, and let 
$\WW = (\W, \wo, \ldots)$ be a fixed 
domain of variables, 
totally ordered by  $\wo$.
%
%
We define three quasi orders on monomials, related by the following refinements:
\[
\xymatrix{
f \lexorder f' \ & \ f  \divorder f' \ \ar@{=>}[r] \ar@{=>}[l] & \ f \wpo f'.
}
\]

The two left ones are standard orders.
Given two monomials $f$ and $f'$,
\emph{division order} is defined by
$f \divorder f'$ if $f \cdot g = f'$ for some $g\in\monomials \WW$.
In order to define \emph{lexicographic order}, 
we consider the $\wo$-largest variable $\w\in\W$ such that $\degpol f \w \neq \degpol {f'} \w$, and
set $f\lexorderstrict f'$ if $\degpol f \w < \degpol {f'} \w$.
Lexicographic order is a \emph{term order}:
%
\begin{lemma} \label{lem:lex is to}
For every monomials $f, f', g\in\monomials \WW$, 
\[
f\divorder f' \implies f \lexorder f' \implies f\cdot g \lexorder f' \cdot g.
\]
\end{lemma}
%
%
The quasi order $\wpo$ is defined as relaxation of division order modulo renamings of variables:
\begin{align}  
\begin{aligned}
\label{eq:wpodef}
f \wpo f' & \iff \prettyexists{\iota \in \E}{\iota(f) \divorder f'}.
\end{aligned}
\end{align}
%
%
We argue that \eqref{eq:wpodef} defines a quasi order, i.e., a transitive relation. 
Indeed, as action of embeddings preserves the division order, namely 
\begin{align*}
f\divorder f' & \implies \kappa(f)\divorder\kappa(f') 
\end{align*}
for every $\kappa \in \E$,
we may deduce transitivity of $\wpo$: 
\[
\iota(f) \divorder f' \ \wedge \ 
\kappa(f') \divorder f'' \ \implies \ 
(\kappa\circ\iota)(f) \divorder \kappa(f') \divorder f''.
\]
$\monomials \WW$ may be identified with $\agepar \WW {\N\setminus 0}$, and $\wpo$ with the 
quasi order determined by $\omega$ on 
$\agepar \WW {\N\setminus 0}$.
Thus 
~
$\wpo$ is a \wqo, if $\WW$ is $\omega$-well structured.
%
%
\begin{definition}
\label{def:lm} 
\rm
Given a non-zero polynomial $f \in \polyring \R \WW$, $f\neq\zeropol$, the \emph{leading monomial} 
$\lm{f} \in \monomials \WW$ of $f$
is the $\lexorder$-largest monomial appearing in $f$;
the \emph{leading coefficient} $\lc f\in \R$ of $f$ is the coeffcient of $\lm f$ in $f$; and
the \emph{leading term} of $f$ is $\lt f = \lc{f} \cdot \lm{f}$.
\end{definition}
\begin{lemmarep} \label{lem:ltcdot}
For every  $f\in\polyring \R \WW$ and $g\in\monomials \R$ it holds
$\lm{g \cdot f} \ = g\cdot \lm f$ and
$\lt{g \cdot f} \ = g\cdot \lt f$.
\end{lemmarep}
%
%
\begin{proof}
Due to Lemma \ref{lem:lex is to}, the $\lexorder$-largest monomial
$\lm f$ in $f$, after mutliplying $f$ by a monomial $g$, is still the $\lexorder$-largest
monomial in $g\cdot f$, namely $\lm {g\cdot f} = g\cdot \lm f$.
As $g$ is a monomial, we have $\lc {g\cdot f} = \lc f$, which implies
$\lt {g\cdot f} = g\cdot \lt f$.
\end{proof}

\para{Characteristic pairs}
Instead of just monomials, we consider monomial-set pairs:
\[
\CP \ = \ \setof{(f,v) \in \monomials \WW \times \pfin {\W}}{\vars f \subseteq v}.
\]
\begin{definition}
\rm
The \emph{characteristic pair} of $f$ is defined as 
$\cc f = (\lm f, \vars f) \in \CP$.
\end{definition}

We extend division and lexicographic orders to $\CP$
by combining them with inclusion (we overload symbols $\divorder$ and $\lexorder$):
\begin{align*}
(f, v) \ \divorder \ (f', v') & \iff f\divorder f' \ \wedge \ v \subseteq v' \\
(f, v)  \lexorderstrict (f', v') & \iff f\lexorderstrict f' \ \wedge \ v \subseteq v'.
\end{align*}
While lexicographic order on monomials does not need to be well founded, its extension to $\CP$ does.
This  is crucial in the sequel.
\begin{lemmarep} \label{lem:lexwf}
The extended $\lexorder$ is well founded.
\end{lemmarep}
\begin{proof}
We need to show that every decreasing chain 
\[
(f_1, v_1)  \lexorderinv (f_2, v_2) \lexorderinv \ldots
\]
eventually stabilises.
The second component eventually
stabilises, i.e., there is $n\in\N$ and $v\subseteqfin \W$ such that for all $i\geq n$ we have
$v_i = v$.
Since $v$ is finite, lexicographic order is well founded on 
$\setof{f \in \monomials \WW}{\vars f \subseteq v}$, and
therefore $f_n \lexorderinv f_{n+1} \lexorderinv \ldots$ eventually stabilises too.
\end{proof}

Similarly we extend $\wpo$ to $\CP$, by combining relaxation of division order \eqref{eq:wpodef}
with inclusion:
\begin{align} \label{eq:wpodef2}
(m, v) \ \wpo \ (m', v')  \iff 
\prettyexists{\iota \in \E}{\iota(m) \divorder m' \wedge \iota(v) \subseteq v'}.
\end{align}

\begin{lemma} \label{lem:extwqo}
The extended $\wpo$ is a \wqo, if $\WW$ is $\omega$-well structured.
\end{lemma}
\begin{proof}
Given a pair $(f, v) \in \CP$, define
 $(\ell_{(f,v)} : v\to\N) \in \agepar \WW \N$ as follows:
\[
\ell_{(f,v)}(\w) \ = \ \begin{cases}
\degpol f \w & \text{ if } \w \in \vars f \\
0 & \text{ if } \w \in v\setminus \vars f.
\end{cases}
\]
We observe that $(f, v) \wpo (f', v')$ exactly when 
$\ell_{(f,v)} \detorder \ell_{(f', v')}$ are related by the quasi order determined by $\omega = (\N, \leq)$.
Therefore $\wpo$ is a \wqo, since $\WW$ is $\omega$-well structured.
\end{proof}

%

\para{Proof of Theorem \ref{thm:equiv noeth}}

Fix a totally ordered and $\omega$-well structured domain $\WW = (\W, \wo, \ldots)$.
For simplicity, we restrict to the case when $\R$ is a \emph{field} --
the general proof for an arbitrary Noetherian commutative ring can be found in the full version
\cite{GL24-arxiv}. 
%
We proceed by a careful adaptation of the classical proof, see e.g.~\cite[Section 2.5]{Cox15}
using $\wpo$ in place of division order, and
characteristic pairs in place of just leading monomials.

Consider an arbitrary ideal $\I \subseteq \polyring \R \WW$, and the set of characteristic pairs
of all polynomials in $\I$:
\[
\cc \I = \setof{\cc f}{f\in\I}.
\]
Let $J\subseteq \CP$ be the upward closure of $\cc \I$ with respect to $\wpo$:
\[
J = \setof{(f,v)\in \CP}{(f',v')\wpo (f,v) \text{ for some } (f',v')\in \cc \I}.
\]
The set $J$ is thus upward-closed and therefore has a finite basis $B\subseteqfin J$
(relying on Lemmas \ref{lem:wqo}  and \ref{lem:extwqo}).
W.l.o.g.~$B\subseteq \cc \I$, as $J$ is the upward closure of $\cc \I$.
Therefore:
\begin{claim} \label{claim:fv}
 For every $(f,v)\in B$, there is a polynomial $g_{f,v}\in\I$  such that
$\cc {g_{f,v}} = (f,v)$.
\end{claim}
Let $\B = \setof{g_{f,v}}{(f,v)\in B}$.
We prove that $\B$ is a basis of $\I$, i.e., $\I = \eqidealgen \B$.

One inclusion is obvious, $\eqidealgen \B \subseteq \I$.
In order to prove the opposite inclusion suppose, towards contradiction, that
$\I \setminus \eqidealgen \B$ is nonempty.
Take any polynomial $g\in \I \setminus \eqidealgen \B$ whose characteristic pair
$(f,v) = \cc g$ is $\lexorder$-minimal, i.e., there is no $g' \in \I\setminus \eqidealgen \B$
with $\cc {g'} \lexorderstrict \cc g$ (cf.~Lemma\ref{lem:lexwf}).
Choose any $(f', v')\in B$ such that 
$(f', v')\wpo (f,v)$, i.e., there is some $h\in\monomials \WW$ and $\iota\in\E$ such that
\[
h\cdot\iota(f') = f \qquad
\iota(v')\subseteq v.
\]
Equivalently, using Claim \ref{claim:fv} and 
taking $g' = g_{f',v'}\in\B$, we have:
\[
h\cdot \iota(\lm {g'}) = \lm g \qquad
\iota(\vars {g'})\subseteq \vars g.
\]
As $\lm{\_}$ and $\vars{\_}$ both commute with the action of embeddings,
the two equalities 
rewrite first to
\begin{align} \label{eq:cmcmcm}
h \cdot \lm{\iota(g')} \ = \ \lm{g} \qquad
\vars {\iota(g')} \ \subseteq \ \vars{g},
\end{align}
and then,
by Lemma \ref{lem:ltcdot}, the first one rewrites further to
\begin{align} 
\label{eq:lm2}
\lm{h \cdot \iota(g')} \ & = \ \lm{g}.
\end{align}
Let $g''$ be the polynomial appearing on the left-hand side:
$
g'' \ :=  h \cdot \iota(g').
$
By \eqref{eq:lm2}, the leading terms appearing on both sides of
\eqref{eq:lm2} cancel out in $g - g''$.
We thus have
\begin{align} \label{eq:lmlex}
\lm {g - g''} \lexorderstrict \lm{g}.
\end{align}
Furthermore, by the first equality in \eqref{eq:cmcmcm} we get $\vars h \subseteq \vars {g}$,
and then by the second one we
get $\vars {g''} \subseteq \vars{g}$, and hence
\begin{align} \label{eq:varssub}
\vars {g - g''} \subseteq \vars{g}.
\end{align}
By \eqref{eq:lmlex} and \eqref{eq:varssub} we thus derive:
\begin{align} \label{eq:contr1}
\cc {g - g''} \lexorderstrict \cc{g}.
\end{align}
Finally, as both $\eqidealgen \B$ and $\I$ are ideals,
knowing that $g'' \in \eqidealgen \B\subseteq \I$ and $g \in \I\setminus \eqidealgen \B$,
we deduce:
\begin{align} \label{eq:contr2}
(g - g'') \in \I\setminus \eqidealgen \B.
\end{align}
Conditions
\eqref{eq:contr1} and \eqref{eq:contr2} are in contradiction with $\lexorder$-minimality of $\cc{g}$,
which completes the proof
of Theorem \ref{thm:equiv noeth}.
\qed

\begin{slexample} \label{ex:Hilb-thm}
The base constructed in the proof of Theorem \ref{thm:equiv noeth} consists of 
(representatives of) polynomials with $\wpo$-minimal characteristic pairs.
For illustration, 
take ordered nonnegative integers as the domain of variables, $\WW = \set{a_0 < a_1 < ...}$, 
and the rational field $\Qf$. 
Consider an ideal $\I \subseteq \polyring \Qf \WW$ containing those polynomials 
   without constant term, where the sum of coefficients of monomials of even degree 
   is equal to the sum of coefficients of monomials of odd degree.
   An example of a base given by the proof of Theorem \ref{thm:equiv noeth}  is 
   \begin{align*}
   \B=\set{a_0^2 + a_0, a_0 - a_1},
   \end{align*}
   as the following two characteristic pairs    are the $\wpo$-minimal elements in $\cc \I$:
   \[
   c_1 = \cc{a_0^2 + a_0} = (a_0^2, \set{a_0})  \qquad c_2 = \cc{a_0 - a_1} = (a_1, \set{a_0, a_1}).
   \]
   Indeed, consider the leading term of a non-zero polynomial $f\in \I$.
   If the degree of some variable in $\lm f$ is at least $2$, we have $c_1 \wpo \cc f$.
   Otherwise all variables appearing in $\lm f$ are of degree $1$.
   As $f\in\I$, it uses more that one variable, and the largest of them must necessarily appear 
   in $\lm f$.
   Therefore $c_2 \wpo \cc f$.
   
   There are several other possible bases that could be chosen in the proof of Theorem \ref{thm:equiv noeth},
   for instance
   $\B_n=\set{a_0^2 + a_0, a_0^n + (-1)^n a_1}$ for any $n\geq 1$.
\end{slexample}

\para{Proof of Theorem \ref{thm:wqo}}

Let $\WW = (\W, \ldots)$ be a relational structure and $\R$ a commutative ring such that
$\polyring \R \WW$ is equivariantly Noetherian.
It is sufficient to prove that $(\monomials \WW, \wpo)$ is a \wqo,
as it is isomorphic to the quasi order determined by
$\omega = (\N, \leq)$ on $\agepar{\WW} {\N\setminus 0}$, and hence also on $\agepar \WW \N$.

Given any infinite sequence of monomials $f_1, f_2, \ldots$, we prove that it contains a dominating pair. 
Let $\I_i 
\subseteq \polyring \R \WW$ 
be the equivariant ideal generated  by the set $\set{f_1, \ldots, f_{i-1}}$.
Clearly, $\I_1 \subseteq \I_2 \subseteq \ldots$.
\begin{claim}
$\I_i = \I_{i+1}$ for some $i\in\N$.
\end{claim}
To prove the claim, let $\I \subseteq \polyring \R \WW$  be the equivariant ideal of polynomials generated by 
the infinite set $\set{f_1, f_2, \ldots}$.
By assumption, it has a finite basis $\B$, $\I = \eqidealgen \B$.
Choose $i\in\N$ such that  $\B \subseteq \I_i$.
Then $\I_i = \eqidealgen \B = \I$ and hence $\I_i = \I_{i+1} = \ldots = \I$.

\smallskip

By the claim we have $f_{i+1} \in \I_{i}$, i.e., 
for some $r_1, \ldots, r_{i}\in\R$, embeddings $\iota_1, \ldots, \iota_{i}\in\E$, 
and monomials $h_1, \ldots, h_{i}\in\monomials \R$,
\[
f_{i+1} \ = \ r_1 \cdot h_1 \cdot \iota_1(f_1) + \ldots + r_{i} \cdot h_{i} \cdot \iota_{i}(f_{i}).
\]
Since $f_{i+1}$ is a monomial, all monomials $h_k \cdot \iota_k(f_k)$ different than $f_{i+1}$ 
necessarily
cancel out, and moreover $f_{i+1}$ has to be equal to $h_k \cdot \iota_k(f_k)$ for some $k\leq i$. 
In other words, we get a dominating pair $f_k \wpo f_{i+1}$,
as required.
\qed


\section{Application: Orbit-finitely generated vector spaces} 
\label{sec:lin}

We use Theorem \ref{thm:equiv noeth} to give a simple proof of 
a generalisation of the key result of \cite{BKM21}:
orbit-finitely generated vector spaces have finite chains of equivariant subspaces.
In the following, $\R$ is a field and $\WW = (\W, \ldots)$ an $\omega$-well structured 
relational structure.

\para{Orbit-finitely generated vector spaces}

For $d\in\N$, consider the vector space $\lin {\otu \W d}$ freely generated by the set $\otu \W d$
of $d$-tuples of elements of $\W$.
Thus $\lin {\otu \W d}$ consists of finite formal sums 
\begin{align} \label{eq:vectors}
r_1 \cdot t_1 + \ldots + r_m \cdot t_m,
\end{align}
where $r_1, \ldots, r_m \in \R$ and $t_1, \ldots, t_m \in \otu \W d$.
Equivalently, $\lin{\otu \W d}$ may be defined as the set of all functions $v : \otu \W d \to \R$
such that $v(t) \neq 0$ for finitely many $t\in\otu\W d$.
Addition of vectors \eqref{eq:vectors} and multiplication by a scalar $r\in\R$ are then defined 
pointwise, as expected.
Embeddings $\iota\in\E$ act 
pointwise on tuples $t = (\w_1, \ldots, \w_d)$, namely
$\iota(t) = (\iota(\w_1), \ldots$, $\iota(\w_d))$, and also pointwise on vectors:
\[
\iota(r_1 \cdot t_1 + \ldots + r_m \cdot t_m) \ = \ 
r_1 \cdot \iota(t_1) + \ldots + r_m \cdot \iota(t_m).
\]
A subset of $\lin{\otu \W d}$ is called \emph{equivariant} if it is invariant under 
action of $\E$.

\begin{slexample}
When $d = 2$, generators are ordered pairs $(\w, \w')\in\otu \W 2$.
A vector in $\lin {\otu \W 2}$ can be presented as a finite directed graph, 
possibly containing self-loops,
whose nodes are elements $\w\in\W$ and edges
are labelled by elements of the field $\R$.
Addition of vectors amounts to merging two such graphs, and adding labels of common edges.
Multiplication by a scalar $r\in\R$ amount to multiplying labels of all edges by $r$.
The set of those graphs where the sum of labels is 0 is an equivariant subspace of $\lin {\otu \W 2}$.
\end{slexample}

We encode vectors as polynomials as follows.
A single generator $t = (\w_1, \w_2, \ldots, \w_d) \in \otu \W d$ is encoded as the monomial
\[
\enc t \ = \ {\w_1}^{1} \cdot {\w_2}^{2} \cdot \ldots \cdot {\w_d}^{2^d}.
\]
As repetitions $\w_i = \w_j$ are allowed, the
degree of a variable $\w$ is thus equal to the sum of powers $2^i$, ranging over 
positions $i$ such that $\w_i = \w$.
This guarantees injectivity of $\enc{\_}$, namely different tuples $t$ have different encodings.
Then a vector $v = r_1 \cdot t_1 + \ldots + r_m \cdot t_m$ is encoded as the polynomial
\[
\enc v \ = \ r_1 \cdot \enc{t_1} + \ldots + r_m \cdot \enc{t_m},
\]
which yields a linear map $\enc \_ : \lin{\otu \W d} \to \polyring \R \WW$.


Vector space $\lin{\otu \W d}$ is indeed \emph{orbit-finitely generated} once $\WW$ is assumed to be
$\omega$-well structured, due to the following lemma:
\begin{lemma}[oligomorphicity%
\footnote{
If we restrict to automorphisms, instead of all embeddings, this is a characterisation of
oligomorphic automophism groups \cite{Cameron}, and $\omega$-categorical structures \cite{H97}.
Furthermore, as mentioned in Section \ref{sec:prelim},
when $\WW$ is homogeneous,
action of embeddings is essentially the same as action of automorphisms.
}]
 \label{lem:of}
$\W^d$ is orbit-finite, namely 
$\W^d = \E(\Bs) \defeq \setof{\iota(t)}{\iota\in\E, t\in\Bs}$ for some  $\Bs\subseteqfin \otu \W d$.
\end{lemma}
Any such finite set $\Bs$ we call a \emph{presentation} of  $\W^d$.
\begin{proof}
Suppose $\otu \W d$ is orbit-infinite.
In such case
there is an infinite sequence $t_1, t_2, \ldots \in \otu \W d$ such that
$\otu \W d = \E(\set{t_1,t_2, \ldots})$ and
$(*)$ there is no $i,j\in\N$ and embedding $\iota\in\E$ such that $\iota(t_i) = t_j$.
We know that $(\monomials \R, \wpo)$ is a \wqo and hence
the infinite sequence of monomials
\[
\enc {t_1}, \ \enc{t_2}, \ \ldots
\]
in $\monomials \R$
admits a domination, namely $\iota(\enc{t_i}) \divorder \enc{t_j}$ for some $i<j$ and $\iota \in \E$.
A crucial fact is that all monomials used in encodings 
have the same degree $D=1 + 2 + \ldots + 2^d = 2^{d+1}-1$,
understood as the sum of degrees of all variables.
Therefore $\iota(\enc{t_i}) \divorder \enc{t_j}$ implies $\iota(\enc{t_i}) = \enc{t_j}$.
The encoding 
commutes with action of embeddings:
$\iota(\enc t) = \enc{\iota(t)}$.
This implies $\enc{\iota(t_i)} = \enc{t_j}$, and by injectivity of $\enc \_$ also
$\iota(t_i) = t_j$, thus contradicting $(*)$.
\end{proof}

\para{Finite length}

Theorem \ref{thm:lin} generalises \cite[Lemma IV.8]{BKM21} which only considers
$\WW = \QQ = (\Q, \leq_\QQ)$. 
On the other hand, \cite{BKM21} provides explicite upper bound on the length.
%
\begin{theorem}[Finite length] 
\label{thm:lin}
Let $\R$ be any field.
If $\WW$ is a reduct of an $\omega$-well structured and totally ordered structure then
strictly increasing chains of equivariant subspaces of $\lin {\otu \WW d}$ are finite.
\end{theorem}
\begin{proof}
Consider any increasing chain of equivariant subspaces
\[
\VS_1 \subseteq 
\VS_2 \subseteq \ldots 
\]
and let $\I_i$ be the ideal in $\polyring \R \WW$ generated
(in the classical sense) by $\enc{\VS_i}$, the encodings
of vectors from $\VS_i$.
As action of embeddings commutes with encoding,
equivariance of subspaces $\VS_i$ implies equivariance of the sets $\enc{\VS_i}$.
As action of embeddings commutes with addition and multiplication of polynomials, 
ideals $\I_i$ are equivariant.
By Theorem \ref{thm:equiv noeth} and Lemma \ref{lem:equiv noeth}, the chain 
$
\I_1 \subseteq \I_2 \subseteq \ldots
$
stabilises, namely, $\I_i = \I_j$ for all sufficiently large $i,j$.
We claim:
\begin{claim}
$\I_i = \I_j$ implies $\enc{\VS_i} = \enc{\VS_j}$.
\end{claim}
\noindent
Having the claim we immediately complete the proof, as 
$\enc{\VS_i} = \enc{\VS_j}$ implies $\VS_i = \VS_j$, by injectivity of $\enc{\_}$.

For demonstrating the claim it is sufficient to prove the inclusion
$\enc{\VS_i} \supseteq \enc{\VS_j}$. 
Given an arbitrary polynomial $f \in \enc{\VS_j}$, we thus prove $f\in\enc{\VS_i}$.
As $\I_i = \I_j$, we have $f\in \I_i$, i.e.,
\begin{align} \label{eq:f}
f \ = \ r_1 \cdot h_1 \cdot \enc{v_1} + \ldots + r_n \cdot h_n \cdot \enc{v_n}
\end{align}
for some vectors $v_1, \ldots, v_n \in \VS_i$, monomials
$h_1, \ldots, h_n \in \monomials \R$ and ring elements $r_1, \ldots, r_n \in \R$.
As all monomials used in encodings have the same degree $D$,
all monomials $g$ appearing in
$\enc{v_1}, \ldots, \enc{v_n}$ have degree $\degmon g = D$, as well as all monomials appearing
in $f$.
This means that all the monomials of degree larger than $D$, 
obtained as a multiplication $h_i \cdot g$, forcedly cancel out in the final sum
\eqref{eq:f}.
Formally ($\zeropol$ is the zero polynomial):
\[
\sum_{i\in I} r_i \cdot h_i \cdot \enc{v_i} \ = \ \zeropol,
\]
where the sum ranges over those $i$ where $h_i$ has positive degree:
$
I = \setof{i\in\setfromto 1 n}{\degmon {h_i} > 0}.
$
All polynomials $r_i \cdot h_i \cdot \enc{v_i}$, for $i\in I$, can be thus safely removed.
As $\degmon {h_i} = 0$ exactly when $h_i$ is a unit monomial, we end up with
\[
f \ = \ \sum_{i\notin I} r_i \cdot \enc{v_i} \ = \ \enc{\sum_{i\notin I} r_i \cdot v_i},
\]
the latter equality following from linearity of $\enc \_$.
Therefore $f \in  \enc{\VS_i}$, as required.
\end{proof}

\para{Weighted Register Automata} \label{sec:zeroness}
Using Theorem \ref{thm:lin}
we obtain a generalisation
of another result of \cite{BKM21}:
decidability of zeroness of weighted register automata.
%


For $d\in\N$, let $Q_d = \otu \W d \cup\{\bot\}$.
The auxiliary element $\bot$ is assumed to be equivariant, namely 
$\iota(\bot) = \bot$ for all $\iota \in \E$.
The vector space $\lin{Q_d}$ is defined like $\lin{\otu \W d}$, 
as the set of finite formal sums of elements of $Q_d$.

Classically, transition function of a weighted automaton with finite state-space $Q$
maps every input symbol and state to a vector over states:
$\Sigma\times Q\to \lin Q$.
Our definition of weighted $d$-register automaton 
(equivalent to the definition of \cite{BKM21}) is along the same lines,
except that input alphabet is $\Sigma=\W$,
state-space is orbit-finite $Q = Q_d$, 
and transition function is equivariant.

A weighted $d$-register automaton consists of a \emph{transition} function
$\delta : \W\times Q_d\to\lin {Q_d}$ and a \emph{final} function
$F : Q_d\to\R$, both equivariant, namely for every $\w\in\W$, $q\in Q_d$,  
and $\iota\in\E$:
\[
\delta(\iota(\w), \iota(q)) \ = \ \iota(\delta(\w, q))
\qquad
F(\iota(\w)) \ = \ F(\w).
\]
$F$ extends uniquely to a linear map $F : \lin {Q_d}\to \R$.
Likewise, for every $\w\in\W$, the function $\delta(\w, \_) : Q_d\to\lin {Q_d}$ extends 
uniquely to a linear map $\widetilde\delta(\w) : \lin {Q_d}\to \lin {Q_d}$.
We overload $\bot$ and use it for vector defined by 
the singleton formal sum $\bot\in\lin{Q_d}$, the \emph{initial} vector.
The semantics of a weighted register automaton 
$(d, \delta, F)$
maps every word $w = \w_1 \ldots \w_\ell\in\W^*$ to the \emph{final} vector
\[
\widetilde\delta(w) \ = \ \widetilde\delta(\w_\ell) \circ \ldots \circ \widetilde\delta(\w_1)(\bot) \ \in \ \lin {Q_d},
\]
on which $F$ is evaluated to get the \emph{value} $F(\widetilde\delta(w))\in\R$ 
assigned by the automaton to $w$.
By equivariance of $\bot$, $\delta$ and $F$, the function
$\widetilde\delta$ is equivariant and
the value is stable under the 
pointwise action of $\E$ on words,
$\iota(a_1 \ldots a_\ell) = \iota(a_1)\ldots \iota(a_\ell)$.
\begin{lemma} \label{lem:value}
$F(\widetilde\delta(\iota(w))) = F(\widetilde\delta(w))$ for every $\iota \in \E$.
\end{lemma}

The \emph{zeroness} problem asks, given a weighted register automaton $(d, \delta, F)$,
whether it is \emph{zeroing}, i.e., whether it assigns value 0 to all input words.
When an automaton is input to an algorithm, 
each of the functions $\delta$, $F$ is given by its values on a
(naturally defined) presentation of its domain.
We prove that zeroness is characterised by a finite testing set:
\begin{theorem}\label{thm:zeroness}
Let $\R$ be any field and
$\WW$ be a reduct of an $\omega$-well structured and totally ordered structure.
For every weighted register automaton there is a finite set $T\subseteq \W^*$ 
such that the automaton is zeroing if and only if
it assigs value 0 to all words in $T$.
\end{theorem}
%
%
%
\begin{proof}
Consider a weighted register automaton $(d, \delta, F)$.
%
For $i \in \N$, let $\otu \W {\leq i} = \otu \W 0 \cup \ldots \cup \otu \W i$, and
$\VS_i\subseteq\lin{Q_d}$ be the vector space spanned by 
$V_i = \widetilde\delta(\otu \W {\leq i}) \defeq \setof{\widetilde\delta(w)}{w\in\otu \W {\leq i}}$.
Let $\Bs_i$ be a presentation of $\otu \W i$, and let
$\Bs_{\leq i} = \Bs_0 \cup \ldots \cup \Bs_i$.
Therefore
\begin{align} \label{eq:VVi}
\otu \W {\leq i} = \E(\Bs_{\leq i}), 
\end{align}
and
by equivariance of $\widetilde\delta$ we have:
\begin{align} \label{eq:Vi}
V_i = \E(\widetilde\delta(\Bs_{\leq i})).
\end{align}
Thus  the sets $V_i$ are equivariant,  and so are the subspaces $\VS_i$.
As an immediate consequence of Theorem \ref{thm:lin} we obtain:
\begin{claimrep} 
\label{claim:linbot}
Strictly increasing chains of equivariant subspaces of $\lin {Q_d}$ are finite.
\end{claimrep}
\begin{proof}
%
We claim that subspaces of $\lin{Q_d}$ split into two types defined below.
Consider an equivariant subspace $\VS\subseteq \lin{Q_d}$ and
its projection $\VS'\subseteq\lin{\otu \W d}$ to $\lin{\otu \W d}$.
If $\VS$ contains a vector $r\cdot\bot$ for some $r\neq 0$ then 
$\VS$ contains vectors $r\cdot\bot$ for all $r\in\R$ and therefore
$\VS = \VS'\times\R$.
Otherwise, every vector $v' \in \VS'$ extends uniquely to a vector in $\VS$.
In both cases, whenever two subspaces $\VS_1, \VS_2\subseteq \lin{Q_d}$ are related by strict inclusion,
their projections $\VS'_1, \VS'_2$ to $\lin{\otu \W d}$ also are:
\begin{align} \label{eq:VSiff}
\VS_1\subset\VS_2 \implies \VS'_1\subset\VS'_2.
\end{align}
Furthermore, a subspace of the first type is never included in a subspace of the second type.
Therefore every strictly increasing chain $\VS_1 \subset \VS_2 \subset \ldots$ 
of equivariant subspaces of $\lin{Q_d}$ splits into two subchains:
a chain of subspaces of the second type, followed by chain of subspaces of the first type.
By \eqref{eq:VSiff} each of the subchains corresponds to a strictly increasing chain
of subspaces of $\lin{\otu \W d}$, and hence is finite by Theorem \ref{thm:lin}.
In consequence, the whole chain is necessarily finite.
\end{proof}
By the claim, the chain 
$
\VS_0  \subseteq  \VS_1  \subseteq \ldots,
$
stabilises at some $n$:
\begin{align} \label{eq:n}
\VS_n  =  \VS_{n+1}  =  \ldots.
\end{align}
Since the transformation $\VS_i \mapsto \VS_{i+1}$ does not depend on $i$, the
stabilisation \eqref{eq:n} is guaranteed once $\VS_n = \VS_{n+1}$.
Zeroness of $(d, \delta, F)$ is then equivalent to
$F(\widetilde\delta(w)) = 0$ for all $w\in\W^{\leq n}$.
Using \eqref{eq:VVi} and Lemma \ref{lem:value} we refine the characterisation:
\begin{claim} \label{claim:zeroness}
Zeroness of $(d, \delta, F)$ is equivalent to
$F(\widetilde\delta(w)) = 0$ for all $w\in\Bs_{\leq n}$.
\end{claim}
Hence $T = \Bs_{\leq n}$ satisfies 
Theorem \ref{thm:zeroness}, which is thus proved.
\end{proof}

\begin{remark} \label{rem:zeroness}
\rm
Under mild computability assumptions on $\R$ and $\WW$ (omitted here for simplicity, but satisfied
by all $\omega$-well structured domains mentioned so far),
this characterisation yields decidability of the zeroness problem.
This generalises \cite[Lemma VIII.1]{BKM21}, but
without providing complexity upper bound.
Indeed, a stabilisation point $n$ \eqref{eq:n} can be computed as follows:
enumerate in parallel, for all $n\in\N$, linear combinations of vectors from $V_n$,
and stop when 
\[
\widetilde\delta(\Bs_{\leq n+1}) \ \subseteq \ V_n.
\]
Indeed, 
by \eqref{eq:Vi}, this implies $V_{n+1} \subseteq V_n$, and hence
$\VS_{n+1}\subseteq\VS_n$.
Once a stabilisation point $n$ is computed, zeroness of an automaton is 
decided by computing the value assigned to all words in $\Bs_{\leq n}$.
\end{remark}


\section{Computability assumptions} \label{sec:comp}

We now formulate computability assumptions on a field $\R$ and structure $\WW$ 
which are necessary in subsequent sections.

A field $\R$ is \emph{computable}, i.e., its elements are finitely represented,
and field operations can be computed using this representation.
So is, for instance, the rational field $\Qf$, as well as every finite field.

The rest of this section is devoted to spelling out computability assumptions on
an $\omega$-well structured structure $\WW = (\W, \ldots)$, in
Definition \ref{def:comp} below.
Given $\Vs\in\age\WW$,
any restriction $\restr\iota \Vs$ of an embedding to $\Vs$ 
we call \emph{local embedding} of $\Vs$.
It is legitimate to use local embeddings instead of embeddings, 
as their action on a polynomial depends only on restriction to its variables:
\begin{lemma} \label{lem:restr}
$\restr \iota {\vars f} = \restr {\iota'} {\vars f} \implies \iota(f) = \iota'(f)$.
\end{lemma}
\noindent
By $\embset \Vs$ we denote the set of all local embeddings of $\Vs$:
\[
\embset \Vs \ = \ \setof{\restr \iota \Vs}{\iota\in\E}.
\] 
For $\Vs, \Vs'\in\age \WW$, we consider the Cartesian product
\[
\embsetprod \Vs {\Vs'} \ \defeq \ (\embset \Vs) \times (\embset {\Vs'}).
\]
We extend to this set the action of embeddings:
for $\kappa \in \E$ and $(\pi,\pi') \in \embsetprod \Vs {\Vs'}$, we put
$\kappa(\pi,\pi') = (\kappa\circ\pi, \kappa\circ\pi')$. 
Clearly, $\embsetprod \Vs  {\Vs'}$ is closed under this action, i.e., it is equivariant.

\begin{lemma} \label{lem:embsetprod}
For $\Vs, \Vs'\in\age \WW$,  $\embsetprod \Vs  {\Vs'}$ is orbit-finite, namely
$\embsetprod \Vs  {\Vs'} = \E(F) \defeq \setof{(\kappa\circ\pi,\kappa\circ\pi')}{\kappa\in\E, (\pi, \pi')\in F}$ for some finite  $F\subseteq \embsetprod \Vs {\Vs'}$.
\end{lemma}
Any such finite set $F$ we call a \emph{presentation} of  $\embsetprod {\Vs}{\Vs'}$.
\begin{proof}
%
%
%
Let $n=\size \Vs$ and $n' = \size{\Vs'}$.
The set $\embsetprod \Vs {\Vs'}$ is essentially a subset of $\otu \WW {n+n'}$.
Formally, having fixed enumerations $\Vs = \set{\v_1, \ldots, \v_n}$,
$\Vs' = \set{\v'_1, \ldots, \v'_{n'}}$, we define an injective mapping
$
m : \embsetprod \Vs {\Vs'} \to \otu \WW {n+n'}
$ as expected:
\[
(\pi, \pi') \ \mapsto \ (\pi(\v_1), \ldots, \pi(\v_n), \pi'(\v'_1), \ldots, \pi'(\v'_{n'})).
\]
The mapping is equivariant, namely commutes with the action of embeddings:
$
m(\iota(\pi, \pi')) \ = \ \iota(m(\pi, \pi')) 
$
for every $\iota\in\E$.
Therefore, equivariance of $\embsetprod \Vs {\Vs'}$ implies equivariance of
$m(\embsetprod \Vs {\Vs'})$. 
By Lemma \ref{lem:of}, the set $\otu \WW {n+n'}$ is orbit-finite.
Importantly, the same proof shows that \emph{every} equivariant subset of $\otu \WW {n+n'}$ is orbit-finite.
Therefore, $m(\embsetprod \Vs {\Vs'}) = \E(F)$ for some finite set $F$.
We use injectivity and equivariance of $m$ to deduce that $m^{-1}(F)$ is a finite presentation of
$\embsetprod \Vs {\Vs'}$.
\end{proof}
%
%

\begin{slexample}\label{eg:pair of emb}
For illustration,
consider the structure ${\omega = (\N, \leq)}$, 
whose embeddings are strictly monotonic maps $\omega\to\omega$, and 
$\Vs=\set{2}$ and $\Vs'=\set{1,3}$.
A presentation of $\embsetprod \Vs {\Vs'}$ is given by $F$ containing five pairs
$(\pi,\pi')$ of local embeddings, corresponding to five cases: 
$\pi(2)$ is smaller
than $\pi'(1)$, or equal to $\pi'(1)$, or between $\pi'(1)$ and $\pi'(3)$, or equal to $\pi'(3)$,
or larger than $\pi'(3)$:
\begin{align*}
\pi_1 : 2 {\mapsto} 2 & \quad \pi'_1 : (1,3){\mapsto}(3,5) & \quad 
\pi_4 : 2 {\mapsto} 3 & \quad \pi'_4 : (1,3){\mapsto}(1, 3) \\
\pi_2 : 2 {\mapsto} 2 & \quad \pi'_2 : (1,3){\mapsto}(2,4) & \quad
\pi_5 : 2 {\mapsto} 4 & \quad \pi'_5 : (1,3){\mapsto}(1, 3) \\
\pi_3 : 2 {\mapsto} 2 & \quad  \pi'_3 : (1,3){\mapsto}(1,3) 
\end{align*}
The following two pairs are not in $F$, for different reasons:
\[
\sigma_1 : 2 {\mapsto} 4 \quad \sigma'_1 : (1,3){\mapsto}(1,4)
\qquad
\sigma_2 : 2 {\mapsto} 3 \quad \sigma'_2 : (1,3){\mapsto}(2,3).
\]
The left pair $(\sigma_1, \sigma'_1)$ belongs to $\embsetprod \Vs  {\Vs'}$ but factorizes through 
$(\pi_4, \pi'_4)$, namely
$(\sigma_1, \sigma'_1) = \kappa(\pi_4, \pi'_4)$ where 
$\kappa(0)=0, \kappa(1)=1$, $\kappa(2)=2$, and $\kappa(n) = n+1$ for $n\geq 3$.
The right pair does not even belong to $\embsetprod \Vs {\Vs'}$, as $\sigma'_2$ does not extend to an embedding.
\end{slexample}

\newcommand{\CA}[1]{$(\textsc{c}_{#1})$}
\begin{sldefinition}\label{def:comp}
We say that an $\omega$-well structured domain $\WW$ 
is \emph{computable} if its elements can be finitely represented,
and the following tasks are computable using this representation:
\begin{enumerate}
\item[\CA 1]
Given $\Vs \in \age \WW$ and $\pi:\Vs\to\W$, decide if  $\pi \in \embset \Vs$.
\item[\CA 2]
Given $\Vs, \Vs' \in \age \WW$, compute a presentation of $\embsetprod {\Vs}{\Vs'}$.
\end{enumerate}
In particular, every finite $\WW$ is computable.
\end{sldefinition}
%
(Our condition \CA 2 is similar to EGB4 in \cite{BD11}.)
%
By reusing the proof of \cite[Proposition 3.4]{HKL18} we can show that the ordinal $\omega$, 
seen as a relational structure, is computable.
In fact the proof can be easily extended to ordinals $\alpha$ which are \emph{effective},
by which we mean that
elements of $\alpha$ (i.e., ordinals smaller than $\alpha$) are finitely represented, and
ordinal order $<$ and addition $+$ are computable using this representation.
For example, the ordinals $\omega$, $\omega^k$, $\omega^{\omega}$ and $\varepsilon_0$ 
are effective, using Cantor Normal Form.
%
\begin{lemmarep} 
Any effective ordinal $\alpha$, seen as a relational structure, is computable.
\end{lemmarep}
\begin{proof}
We start by considering $\omega$;
later we explain how to extend the proof to any effective ordinal.
Let $\Vs\subseteqfin\omega$ be a nonempty subset of $\omega$.
A function $\pi:\Vs\to\omega$
extends to an embedding $\omega\to\omega$ if and only if the following conditions hold:
\begin{align} \label{eq:CA1}
\begin{aligned}
& \v \leq \pi(\v) && \text{ where } \v=\min\Vs \\
& \v_2 - \v_1 \leq \pi(\v_2) - \pi(\v_1)  && \text{ for every } \v_1\leq\v_2.
\end{aligned}
\end{align}
This implies \CA 1.

\smallskip

Concering \CA 2,  given $\Vs, \Vs' \in\age\WW$,
we compute $\embsetprod \Vs{\Vs'}$ by a dynamic algorithm with respect to the orders on 
$\Vs$ and $\Vs'$:
for every two initial segments $\Us \subseteq \Vs$ and
$\Us' \subseteq \Vs'$, we compute
$\embsetprod \Us{\Us'}$, as follows
(the algorithm is similar to \cite[Proposition 3.4]{HKL18}).


The trivial set $\embsetprod {\emptyset} {\emptyset} = \set{(\emptyset,\emptyset)}$ 
contains just one pair of empty functions.
Consider two initial segments $\Us \subseteq \Vs$ and
$\Us' \subseteq \Vs'$.
Assume $\embsetprod {\overline\Us}{\overline\Us'}$ is already computed for all pairs of initial segments 
smaller than 
$\Us, \Us'$,
i.e. for all pairs
$(\overline\Us,\overline\Us') \neq (\Us,\Us')$ of initial segments $\overline\Us\subseteq \Us$ 
and $\overline\Us'\subseteq \Us'$.
For computing $\embsetprod \Us{\Us'}$ we use $\overline\Us\subseteq \Us$ obtained from $\Us$ 
by removing the largest element $\v$ (assuming $\Us$ is nonempty),
and $\overline\Us'\subseteq \Us'$ obtained from $\Us'$ by removing the largest element $\v'$
(assuming $\Us'$ is nonempty).
We compute $\embsetprod \Us{\Us'}$ by applying the following three procedures
(note that their pre-conditions are not disjoint):
\begin{itemize}
\item
Suppose both $\Us$ and $\Us'$ are nonempty.
For every $(\bar\pi, \bar\pi') \in \embsetprod {\overline\Us}{\overline\Us'}$ we extend $\bar\pi$ and $\bar\pi'$ 
to $\pi : \Us \to \omega$ and $\pi' : \Us'\to\omega$ by mapping both $\v$ and $\v'$ to
the minimal
value $\pi(\v) = \pi'(\v')\in\omega$ satisfying the condition \eqref{eq:CA1} for both $\pi$ and $\pi'$.
We add all so obtained pairs $(\pi, \pi')$ to $\embsetprod \Us{\Us'}$.
\item
Suppose $\Us$ is nonempty (while $\Us'$ is either empty or not).
For every $(\bar\pi, \pi') \in \embsetprod {\overline\Us}{\Us'}$ we extend 
$\bar\pi : \overline\Us \to\omega$ to $\pi : \Us \to \omega$
by mapping $\v$ to the minimal value $\pi(\v)$ outside of the range $\pi'(\Us')$ that satisfies the condition \eqref{eq:CA1}.
We add all so obtained pairs $(\pi, \pi')$ to $\embsetprod \Us{\Us'}$.
\item
Supposing $\Us'$ is nonempty  (while $\Us$ is either empty or not), we proceed symmetrically to the previous case.
\end{itemize}
This completes the procedure to compute $\embsetprod \Vs{\Vs'}$.
We have thus completed the proof for $\omega$.
\smallskip

\label{para:comp ord}
The proof extends to any ordinal $\alpha$, as long as
ordinal subtraction is computable in $\alpha$, where subtraction of
ordinals $\beta - \gamma$, for $\gamma < \beta < \alpha$,
is defined as the unique ordinal $\eta$ such that 
$\beta = \gamma + \eta$ (the order of addition is important since addition of ordinals is not commutative).
For any effective ordinal $\alpha$ and $\gamma<\beta<\alpha$, subtraction $\beta - \gamma$
is computable by enumerating ordinals $\eta<\alpha$ and testing whether $\beta=\gamma + \eta$.
\end{proof}
\begin{lemmarep}\label{lem:lex comp}
Lexicographic product of two computable domains is computable.
\end{lemmarep}
\begin{proof}
Let $\WW = (\W, \ldots)$ and $\VV = (\Vs, \ldots)$ be two computable, $\omega$-well structured domains.
We will show that $\WW \domprod {\VV}$ is also computable.
%
%

First, assuming that $\WW$ and $\VV$ satisfy \CA{1},
we show that $\WW \domprod {\VV}$ satisfies \CA{1} as well.
For $\Bs \in \age {\WW \domprod {\VV}}$ let $\overline\Bs\in\age\WW$ be its projection
on the first coordinate.
Moreover, for $\w\in\overline\Bs$, let $\Bs_\w = \setof{\w'\in\W'}{(\w, \w')\in\Bs} \in \age{\VV}$.
Given $\Bs\in\age\WW$ and $\pi : \Bs\to\W\times\W'$, 
we need to check if $\pi$ extend to an embedding of $\WW \domprod {\VV}$.
W.l.o.g.~assume that $\pi$ preserves equality in first coordinate since (this necessary condition can be easily checked).
Then, $\pi$ induces a map $\overline \pi : \overline \Bs \to \W$ and maps $\pi_\w : \Bs_\w  \to \W'$ for 
$\w \in \overline\Bs$.
By Lemma \ref{lem:prod embed}, $\pi$ extends to an embedding of $\WW \domprod {\VV}$ 
only if $\overline\pi$ extends to an embedding $\overline\iota$ of $\WW$ and for every 
$\w\in\overline\Bs$, the map $\pi_\w$ extend to an embedding $\iota_a$ of $\VV$.
Moreover if both these conditions are satisfied, 
then $\pi$ does extend to the embedding of $\WW \domprod {\VV}$ determined by $\overline\iota$ and 
the family of embeddings $(\iota_\w \in \embed{\VV})_{\w\in\W}$, where
$\iota_\w$ is arbitrary for $\w\notin\overline\Bs$.
Since $\WW$ and $\VV$ satisfy \CA{1},
the conditions can be checked and hence $\WW \domprod \VV$ satisfies \CA{1}.
\smallskip

Now, assuming further that $\WW$ and $\VV$ satisfy \CA{2},
we show that $\WW \domprod {\VV}$ satisfies \CA{2} as well, i.e.,
given $V,W \in \age{\WW \domprod {\VV}}$, 
we show how to compute a presentation $E$ of $\embsetprod{V}{W}$.
Let $X\in \age \WW$ be the union of projections of $V$ and $W$ to the first coordinate:
\[
X = \overline{V}\cup\overline{W}
= \setof{a\in\W}{(a,b)\in V\cup W\text{ for some }b\in\Vs}.
\]
Likewise, let $Y\in\age \VV$ be the union of projections of $V$ and $W$ to the second coordinate:
\[
Y = \bigcup_{a\in X} (V_a\cup W_a) = \setof{b\in\Vs}{(a,b)\in V\cup W\text{ for some }a\in\W}.
\]
Above we assume $V_a = \emptyset$ when $a\notin\overline{V}$,
and likewise for $W_a$.
Let $Z = X\times Y \in \age{\WW \domprod {\VV}}$.
\begin{claim}\label{clm:embsetprod restr}
If $F$ is a presentation of $\embsetprod{Z}{Z}$,
then
\[
E := \setof{(\restr{f}{V},\restr{g}{W})}{(f,g)\in F}
\]
is a presentation of $\embsetprod{V}{W}$.
\end{claim}
\begin{claimproofnew}
Consider $(p,q)\in \embsetprod{V}{W}$.
There exists $(p',q')\in \embsetprod{Z}{Z}$ such that $\restr{p'}{V} = p$ and $\restr{q'}{W} = q$.
There exists $(f,g)\in F$ and $\iota \in \embed{\WW \domprod {\VV}}$ such that
\[
\iota \circ f = p' \qquad \text{ and } \qquad \iota \circ g = q' \ .
\]
Which implies
\[
\iota \circ \restr{f}{V} = p \qquad \text{ and } \qquad \iota \circ \restr{g}{W} = q \ .
\]
As $(p,q)\in \embsetprod{V}{W}$ is arbitrary, the claim is proved.
\end{claimproofnew}

It remains to show how to compute a presentation $F$ of $\embsetprod{Z}{Z}$.
As $\WW$ and $\VV$ satisfy \CA{2}, we
can compute presentations $S$ and $T$ of $\embsetprod{X}{X}$ and $\embsetprod{Y}{Y}$ respectively.
Before progressing further we notice that for any set $U$ we have
$\embsetprod{U}{\emptyset} = \set{(\id{U},\id{\emptyset})}$ and
$\embsetprod{\emptyset}{U} = \set{(\id{\emptyset},\id{U})}$.
Let $s = (s_1,s_2)$ be an arbitrary element in $S$.
Let $s(X) := s_1(X)\cup s_2(X)$.
We define a family of finite subsets $I_s$ indexed by $s(X)$ as follows:
\[
I_s(a) =
\begin{cases}
T & \text{when } a\in s_1(X)\cap s_2(X) \\
\set{(\id{Y},\id{\emptyset})} & \text{when } a\in s_1(X)\setminus s_2(X) \\
\set{(\id{\emptyset},\id{Y})} & \text{when } a\in s_2(X)\setminus s_1(X).
\end{cases}
\]
We put:
\[
I_s:= \prod_{a\in s(X)} I_s(a).
\]
For every $t = (t^a_1,t^a_2)_{a\in s(X)} \in I_s$ we define
$f^{s,t} = (f^{s,t}_1,f^{s,t}_2) \in \embsetprod{Z}{Z}$ as:
\[
f^{s,t}_i(a,b) = (s_i(a),t^{s_i(a)}_i(b)).
\]
We now prove that
\[
F := \setof{f^{s,t}}{s\in S,\ t\in I_s}
\]
is a presentation of $\embsetprod{Z}{Z}$.
%
Towards this aim,
consider an arbitrary $u = (u_1,u_2) \in \embsetprod{Z}{Z}$.
The two local embeddings
$u_1,u_2$ induce maps $\overline{u}_1,\overline{u}_2\in\restr{\embed{\WW}}{X}$ and hence
$(\overline{u}_1,\overline{u}_2)\in\embsetprod{X}{X}$.
Since $S$ is a presentation of $\embsetprod{X}{X}$,
there exists $\overline{j}\in\embed{\WW}$ and $s = (s_1,s_2)\in S$ such that
\[
\overline{j} \circ s_1 = \overline{u}_1
\qquad \text{and} \qquad
\overline{j} \circ s_2 = \overline{u}_2 \ .
\]
This implies the following equalities:
\begin{align*}
\overline{j}(s_1(X)\cap s_2(X)) & =  \overline{u}_1(X)\cap\overline{u}_2(X) \\
\overline{j}(s_1(X)\setminus s_2(X)) & =  \overline{u}_1(X)\setminus\overline{u}_2(X) \\
\overline{j}(s_2(X)\setminus s_1(X)) & =  \overline{u}_2(X)\setminus\overline{u}_1(X).
\end{align*}
For every  $a \in s(X)$, the local embeddings
$u_1$ and $u_2$ induce local embeddings $u^a_1$ and $u^a_2$ in
$\restr{\embed{\VV}}{Y}$ and $\restr{\embed{\VV}}{Y}$,
respectively.
Since $T$ is a presentation of $\embsetprod{Y}{Y}$,
there exists $(p^a_1,p^a_2)\in T$ and $j_a \in \embed{\VV}$ such that
\[
j_a \circ p^a_1 = u^a_1
\qquad \text{and} \qquad
j_a \circ p^a_2 = u^a_2 \ .
\]
Define $t \in I_s$ as
\[
(t^a_1,t^a_2) =
\begin{cases}
(p^a_1,p^a_2) & \text{when } a\in s_1(X)\cap s_2(X) \\
(\id{Y},\id{\emptyset}) & \text{when } a\in s_1(X)\setminus s_2(X) \\
(\id{\emptyset},\id{Y}) & \text{when } a\in s_2(X)\setminus s_1(X) \ .\\
\end{cases}
\]
We finish the proof by showing that there exists $\iota\in\embed{{\WW}\domprod{\VV}}$ such that
\begin{equation}\label{eq:prod factor}
\iota \circ f^{s,t}_1 = u_1
\qquad \text{and} \qquad
\iota \circ f^{s,t}_2 = u_2 \ .
\end{equation}
For $a\in s(X)$ there exists $v^a_1,v^a_2\in\embed{\VV}$ such that
\[
\restr{v^a_1}{X} = u^a_1
\qquad \text{and} \qquad
\restr{v^a_2}{X} = u^a_2 \ .
\]
Define a family of embeddings $\iota_a$ indexed by $a\in\WW$ as
\[
\iota_a =
\begin{cases}
j_a & \text{when } a\in s_1(X)\cap s_2(X) \\
v^a_1 & \text{when } a\in s_1(X)\setminus s_2(X) \\
v^a_2 & \text{when } a\in s_2(X)\setminus s_1(X) \\
\id{Y} & \text{otherwise.}
\end{cases}
\] 
Using Lemma \ref{lem:prod embed} we can determine $\iota\in\embed{{\WW}\domprod{\VV}}$ by $\overline{j}$ and the family of embeddings $\iota_a$.
It is routine to verify that $\iota$ satisfies \eqref{eq:prod factor}.
\end{proof}


\section{\Gr bases} \label{sec:Gr}

In this section we set up the background for deciding the ideal membership problem
in Section \ref{sec:decid}.
%
Let $\R$ be a  computable field. 
Let $\WW=(\W, \wo, \ldots)$ be a computable,
$\omega$-well structured and \emph{well ordered}  domain
(we thus strengthen the total order assumption).
Therefore, the lexicographic order on monomials is well founded.

%
%
%
%
%

%

\para{Division step}
%
Let $\G \subseteq \polyring\R \WW \setminus\{\zeropol\}$ be a  set of nonzero polynomials.
Given polynomials $f,f' \in \polyring \R \WW$,
we write
\begin{align}\label{eq:red}
\red f \G f'
\end{align}
if $f'$ is the remainder of \emph{division} of $f$ by some polynomial $\iota(g)$ from $\E(\G)$.
Formally, we define \eqref{eq:red} to mean that 
\begin{align} \label{eq:fr}
f = r\cdot h \cdot \iota(g) + f'
\end{align}
for some $r\in \R, h \in \monomials \WW$, $\iota \in \E$ and $g \in \G$ such that
the leading monomial 
$\lm {r\cdot h\cdot \iota(g)}$ (equal to $h\cdot \iota(\lm g)$)
appears in $f$ with the same coefficient as in $r\cdot h\cdot \iota(g)$
(namely $r\cdot \lc g$).
The monomial $\lm {r\cdot h\cdot \iota(g)}$ 
is called the \emph{head} of the division step \eqref{eq:red}.
The head is the largest monomial that disappears from $f$, i.e., does not appear any more
in $f' = f - r\cdot h\cdot\iota(g)$.
All monomials in $f$ lexicographically larger than the head appear also in $f'$
with the same coefficients, while
some monomials lexicographically smaller than the head 
may be added or removed by a division step.
%

Recall that $\eqidealgen\G$ denotes the equivariant ideal generated by $\G$. 
Using equality \eqref{eq:fr} we deduce:
\begin{lemma} \label{lem:red pres G}
Division step preserves membeship in $\eqidealgen\G$: 
if $\red f \G {f'}$ then $f\in\eqidealgen\G \iff f'\in\eqidealgen\G$.
\end{lemma}
%

%
%

A polynomial $f$ is said to be \emph{reduced} with respect to $\G$ if there is no 
$f' \in \polyring{\R}\WW$ such that $\red f \G f'$.
\begin{lemma}\label{lem:red terminates}
Let 
$\G \subseteq \polyring \R \WW$ be a finite set of nonzero polynomials. 
Every sequence of division steps
\begin{align} \label{eq:inf red}
\red {f_1} \G {\red {f_2} \G {\red {f_3} \G \ldots}}
\end{align}
terminates, i.e., ends in a reduced polynomial.
\end{lemma}
%
%
\begin{proof} 
%
By \emph{reduction ab absurdum}:
assuming an infinite sequence of division steps \eqref{eq:inf red}, we deduce that
the lexicographic order $\lexorder$ on monomials $\monomials \WW$ is not well founded.
%
%

Given an infinite sequence of division steps \eqref{eq:inf red},
we construct an infinite forest $\mathcal F$, whose nodes are labelled by monomials,
as the limit of an increasing sequence of finite forests ${\mathcal F}_1, {\mathcal F}_2, \ldots$.
The construction will maintain the invariant that all monomials appearing in $f_i$ are
among labels of leaves in ${\mathcal F}_i$ (we \emph{mark} these leaves, 
for distinguishing them from other leaves).
The forest ${\mathcal F}_1$ consists just of (finitely many) roots, 
labeled by monomials appearing in $f_1$, all of them marked.
Given ${\mathcal F}_i$, we define ${\mathcal F}_{i+1}$ by adding new children to the 
marked leaf $v$ labelled by the head monomial $m$ of the division step $\red {f_1} \G {f_{i+1}}$.
The leaves are labelled by monomials that are added by the division step, and all are marked.
The node $v$ itself becomes unmarked (even if it remains a leaf), and likewise do all other 
leaves labelled by monomials of $f_i$ that do not appear in $f_{i+1}$.

By the construction of $\mathcal F$, every two monomials $m$, $m'$ related by an edge are in 
strict lexicographic order: $m \lexorderstrictinv m'$.
As the sequence \eqref{eq:inf red} is infinite, the forest is also infinite and, by K{\"o}nig's lemma,
forcedly contains an infinite branch. 
This implies existence of an infinite strictly decreasing chain of monomials
with respect to lexicographic order, which is a contradiction.
%
\end{proof}

We write $\redstar f \G f'$ if there is some finite sequence of division steps from $f$ to $f'$
and $f'$ is reduced.

\para{\Gr basis}
A set $\G \subseteq \polyring \R \WW$ is said to be a \emph{\Gr basis} 
if for every $f \in \eqidealgen \G$ there exists 
$g \in \G$ such that $\lm{g} \wpo \lm{f}$.

\begin{lemma} \label{lem:Gr iff}
If $\G\subseteq \polyring \R \WW$ is a \Gr basis then for every polynomial $f\in\polyring \R \WW$, the following conditions are equivalent: \ 
(1) $f\in \eqidealgen \G$; \ 
(2) $\redstar f\G \zeropol$; \ 
(3) every sequence $\redstar f\G f'$ ends in $f'=\zeropol$.
\end{lemma}
\begin{proof}
The implication (2)$\implies$(1) holds for any $\G$.
Indeed, a sequence $\redstar f\G \zeropol$  of division steps 
yields a decomposition of $f$ into a
sum:
\[
f \ = \ r_1 \cdot h_1 \cdot \iota_1(g_1) \ + \ \ldots \ + \ r_n \cdot h_n \cdot \iota_n(g_n),
\]
where $r_1, \ldots, r_n\in\R$, $g_1, \ldots, g_n\in\G$, $h_1, \ldots, h_n\in\polyring \R \WW$ and $\iota_1, \ldots, \iota_n\in \E$,
which proves membership of $f$ in $\eqidealgen\G$.

For (1)$\implies$(3),
suppose $f\in\eqidealgen\G$, and consider any sequence of division steps starting in $f$.
By Lemma \ref{lem:red terminates} the sequence is necessarily finite, and hence ends in some 
reduced polynomial $f'$.
By Lemma \ref{lem:red pres G}, $f'\in\eqidealgen\G$.
However, as $\G$ is a \Gr basis,
no nonzero polynomial in $\eqidealgen\G$ is reduced, and hence $f'=\zeropol$.

The last implication (3)$\implies$(2) holds vacuously, as every polynomial $f$ admits some sequence
of division steps $\redstar f \G f'$.
\end{proof}
%


\begin{remark} \label{rem:Grob}
\rm
Note that the proof of Theorem \ref{thm:equiv noeth} shows that every equivariant ideal has a finite \Gr basis.
In the sequel we provide a way to \emph{compute} such a basis.
\end{remark}

%

%
\para{Buchberger's algorithm}

By the least common multiple of two monomials $f,f'\in\monomials \WW$,
denoted as $\lcm {f}{f'}$, we mean
the least  (with respect to divisibility) monomial $g$ divisible by both $f$ and $f'$, i.e.~satisfying 
$f \divorder g$ and $f' \divorder g$.

\begin{definition}[$S$-polynomials]\label{def:S-poly}\rm
Given two distinct polynomials $f,g \in \polyring{\R}{\WW}$, $f\neq g$,
we define their \emph{S-polynomial} $\Sof {f,g}$ as
\[
\Sof{f,g} \ := \ \left(\frac{h}{\lt{f}}\cdot f\right) -
\left(\frac{h}{\lt{g}}\cdot g\right),
\]
where $h = \lcm {\lm f} {\lm g}$.
We say that monomial $h$ is \emph{cancelled out} in $\Sof{f,g}$.
For a subset $\G\subseteq{\polyring \R \WW}$ of polynomials, we put:
\begin{align*}
\Ssub \G \ = \ 
\bigcup_{(f,g)\in\E(\G),f \neq g} \! \! \! \! \!\!\! \Sof{f,g}.
\end{align*}
Expanding $\E(\G) = \setof{\iota(g)}{\iota\in\E,\ g\in\G}$, we get:
\begin{align*}
\Ssub \G \ = \ 
\setof{\Sof{\iota(f), \kappa(g)}}{\iota,\kappa\in\E, \ f,g\in\G, \ \iota(f)\neq\kappa(g)}.
\end{align*}
\end{definition}
%

As action of embeddings commutes with S-polynomials, namely 
$\Sof {\kappa(f), \kappa(g)} \ = \ \kappa(\Sof {f,g})$
for every $\kappa\in\E$ and $f,g\in\polyring \R \WW$,
the set $\Ssub \G$ is closed under the action of $\E$, i.e.~is equivariant:
$\Ssub\G \ = \ \E(\Ssub\G)$.
Using assumption \CA 2 we derive the following key fact:
%
%
\begin{lemma}[Computable presentation] \label{lem:min finite}
For every finite set of polynomials $\G \subseteq \polyring{\R}{\WW}$, the set $\Ssub\G$ 
is orbit-finite, namely
\[
\Ssub\G \ = \ \E(\Pres\G) \ \defeq \ \setof{\iota(f)}{\iota\in\E, f\in\Pres\G}
\]
for some presentation $\Pres\G \subseteqfin \Ssub\G$, which is moreover computable.
\end{lemma}
\begin{proof}
We start by showing orbit-finiteness of $\Ssub\G$.
As $\G$ is finite and $\Ssub \G = \bigcup_{g,g'\in\G} \Ssub {g,g'}$, where
\[
\Ssub {g,g'} \ = \  \setof{\Sof{\iota(g),\iota'(g')}}{\iota, \iota' \in \E, \ \iota(g)\neq\iota'(g')},
\]
it is enough to prove that $\Ssub {g,g'}$ is orbit-finite for every two fixed polynomials $g,g'\in\G$.
Let $\Vs = \vars g$ and $\Vs' = \vars {g'}$, and
let $F_{g,g'}$ be a finite presentation of $\embsetprod {\Vs} {\Vs'}$
(recall Lemma \ref{lem:embsetprod}).
Relying on Lemma \ref{lem:restr},
for a local embedding $\pi\in\embset\Vs$ we write $\pi(g)$ to mean the renaming of $g$
by \emph{any} extension of $\pi$ to an embedding of $\WW$.
Likewise we write $\pi'(g')$, for $\pi'\in\embset{\Vs'}$:
\[
\Ssub {g,g'} \ = \  \setof{\Sof{\pi(g),\pi'(g')}}{(\pi, \pi') \in \embsetprod \Vs {\Vs'}, \ \pi(g)\neq\pi'(g')}.
\]
Finally, we define $\Pres{g,g'}\subseteq \Ssub {g,g'}$ as the subset of $\Ssub {g,g'}$ obtained by restricting 
to local pairs of embeddings $(\pi, \pi')$ belonging to $F_{g,g'}$:
\[
\Pres{g,g'} \ = \ \setof{\Sof {\pi(g), \pi'(g')}}{(\pi,\pi') \in F_{g,g'}, \ \pi(g)\neq\pi'(g')}.
\]
As $F$ is finite, $\Pres{g,g'}$ is finite as well, 
and orbit-finiteness of $\Ssub {g,g'}$ follows once we prove that
$\Pres{g,g'}$ is a presentation of $\Ssub{g,g'}$:
\begin{claim}
$\Ssub {g,g'} = \E(\Pres{g,g'})$.
\end{claim}
\begin{claimproofnew}
For proving the inclusion
$\E(\Pres{g,g'}) \subseteq \Ssub {g,g'}$,
we take any polynomial $\kappa(f)\in \E(\Pres{g,g'})$, where $\kappa \in \E$ and $f = \Sof {\pi(g), \pi'(g')} \in \Pres{g,g'}$,
and argue that $\kappa(f) \in \Ssub {g,g'}$.
As action of embeddings commutes with S-polynomials, 
we have
\[
\kappa(\Sof {\pi(g), \pi'(g')}) \ = \ \Sof {\kappa(\pi(g)), \kappa(\pi'(g'))}.
\]
As $\kappa$ is injective, 
knowing $\pi(g)\neq\pi'(g')$, we also know $\kappa(\pi(g))\neq\kappa(\pi'(g'))$,
and therefore $\kappa(f) \in \Ssub {g,g'}$.

The opposite inclusion $\Ssub {g,g'} \subseteq \E(\Pres{g,g'})$ is shown similarly.
We take any polynomial $\Sof {\pi(g), \pi'(g')} \in \Ssub {g,g'}$.
As $F_{g,g'}$ is a presentation of $\embsetprod \Vs{\Vs'}$, there is some 
$(\overline\pi, \overline\pi')\in F_{g,g'}$ and
$\kappa\in\E$ such that
$(\pi,\pi') = (\kappa\circ\overline\pi, \kappa\circ\overline\pi')$.
Knowing $\kappa(\pi(g))\neq\kappa(\pi'(g'))$, we also know $\pi(g)\neq\pi'(g')$,
and therefore $\Sof {\overline\pi(g), \overline\pi'(g')} \in \Pres{g,g'}$.
As action of embeddings commutes with S-polynomials, we get
\[
\Sof {\pi(g), \pi'(g')} \ = \ 
\kappa(\Sof {\overline\pi(g), \overline\pi'(g')}),
\]
and in consequence $\Sof {\pi(g), \pi'(g')}\in\E(\Pres{g,g'})$.
\end{claimproofnew}


Union of all the sets $\Pres{g,g'}$ yields therefore a presentation of $\Ssub \G$:
\[
\Pres{\G} \ \defeq \ \bigcup_{g,g'\in\G}\Pres{g,g'}.
\]
It is moreover computable:
for every $g,g'\in\G$ the finite set $F_{g,g'}$ is computable, by assumption \CA 2,
and therefore the set $\Pres{g,g'}$ is computable as well.
\end{proof}

\begin{lemma}[Orbit-finite \Bu criterion] \label{lem:equivariant buch cri}
A finite set $\G \subseteq \polyring{\R}{\WW}$ is a \Gr basis 
if and only if for all $h \in \Ssub{\G}$,
\[
\redstar h \G \zeropol.
\]
\end{lemma}
\begin{proof}
According to our definition, 
a finite set $\G\subseteq \polyring \R \WW$ of polynomials is a \Gr basis
if and only if the \emph{orbit-finite} set $\E(\G)$
is a \Gr basis in the classical sense.
As the set $\Ssub \G$ contains exactly $S$-polynomials of pairs of polynomials from $\E(\G)$,
we reuse the proof of classical \Bu criterion 
(which is still valid for infinite sets,
see e.g.~\cite[Chapter~2]{Cox15})
to deduce the lemma.
(This is actually the only step which requires the stengthening of assumptions, i.e.,
well ordered $\WW$.)
\end{proof}

\begin{lemma}[Finite \Bu criterion] \label{lem:effective buch cri}
A finite set $\G \subseteq \polyring{\R}{\WW}$ is a \Gr basis 
if and only if for all $h \in \Pres{\G}$,
\[
\redstar h \G \zeropol.
\]
\end{lemma}
\begin{proof}
The 'only if' direction follows immediately by Lemma \ref{lem:equivariant buch cri}.
For the 'if' direction, suppose that $\red h \G \zeropol$ for all $h \in \Pres{\G}$.
Consider any $h' \in \Ssub \G$, i.e.~$h' = \iota(h)$ for some $\iota\in\E$ 
and $h\in \Pres \G$.
By assumption, $\redstar h \G \zeropol$.
We observe that the action of $\E$ (i.e.,~renaming of variables) preserves division steps, namely
\[
\red f \G f' \implies \red {\iota(f)} \G {\iota(f')},
\]
and therefore $\redstar {\iota(h)} \G {\iota(\zeropol) = \zeropol}$.
As $h'\in\Ssub\G$ is arbitrary, 
using Lemma \ref{lem:equivariant buch cri} we conclude that
$\G$ is a \Gr basis.
\end{proof}

Algorithm \ref{algo:buch} is an adaptation of the classical \Bu algorithm to our setting.
The adaptation amounts to using only polynomials from 
$\Pres \G$ instead
of all S-polynomials.

%


\begin{algorithm}[H]
\caption{\textsc{Equivariant \Bu algorithm}}
\label{algo:buch}
\begin{algorithmic}[0]
\State \textbf{Input:} \ A finite set of polynomials $\G\subseteqfin \polyring \R \WW$.
\medskip
\Repeat
\smallskip
\For{$f \in \Pres \G$}
\State compute some reduced $h\in\polyring\R\WW$ such that $\redstar f \G h$ 
\If{$h\neq\zeropol$}
 $\G \gets \G\cup \set{h}$ 
\EndIf 
\EndFor
\Until{$\G$ stabilizes}
\medskip
\State \textbf{Output:} \ $\G$
\end{algorithmic}
\end{algorithm}                                                                                                                                                                  

%
%


\noindent
The above description is indeed an algorithm as, relying on our computability assumption \CA 1, we prove that
one can compute some division step from a given polynomial $f$: 
%

\begin{lemma}[Computability] \label{lem:alg}
Given $f\in\polyring\R\WW$ and $\G\subseteqfin\polyring\R\WW$,
one can check if $f$ is reduced wrt.~$\G$ and,
if $f$ is not, compute some $f'\in\polyring\R\WW$ such that $\red f \G f'$. 
\end{lemma}
\begin{proof}
For every polynomial $g\in\G$ and monomial $m$ of $f$,
we compute the set $D_{g,m}$ of those 
functions $\pi : \vars {\lm g} \to \vars m$ 
that extend to an embedding in $\E$
(which is checked using assumption \CA 1) and satisfy $\pi(\lm g) \divorder m$.
If $D_{g,m} = \emptyset$ for all $g$ and $m$, the polynomial $f$ is reduced wrt.~$\G$.
Otherwise we pick up arbitrary function $\pi  \in D_{g,m}$,
and compute some its extension
 $\pi' : \vars g \to \W$ to $\vars g$ that still extends
to an embedding in $\E$.
This is doable by enumerating all (countably many) candidates $\pi':\vars g \to \W$ and checking them
(again relying on assumption \CA 1).
Termination is guaranteed as $\pi$ does extend to an embedding.
Finally, we divide 
$m$ by $\pi'(\lm g)$ which yields a monomial $h$ satisfying
$
m \ = \ h  \cdot \pi'(\lm g),
$
and compute the remainder $f' = f - r\cdot {\lc g}^{-1} \cdot h\cdot \pi'(g)$, where
$r\in\R$ is the coefficient of $m$ in $f$.
\end{proof}

The algorithm is nondeterministic, namely it may choose different 
sequences of division steps $\redstar f \G h$ and hence different reduced polynomials $h$, 
and therefore it may have many different runs for the same input $\G$.
Nevertheless, we prove termination of its every run, and correctness of output.

\begin{lemma}[Termination]\label{lem:buch stops}
Every run of Algorithm \ref{algo:buch} terminates on every input $\G\subseteq\polyring \R \WW$.
\end{lemma}
\begin{proof}
Towards contradiction, suppose some run of the algorithm is infinite.
Let $\G_n \subseteq \polyring \R \WW$ 
be the value of $\G$ after $n$ iterations of the while loop.
Let $\M_n = \setof{\lm{g}}{g \in \G_n} \subseteq \monomials \WW$.
For every $n$ we have $\M_n \subseteq \M_{n+1}$.
As every polynomial $h$ added in $(n+1)$th iteration is reduced with respect to $\G_{n}$,
we have $\lm{h} \notin \M_n$, and hence
$\M_n \subsetneq \M_{n+1}$ for every $n$.
Therefore we can form an infinite sequence by picking up, for every $n$, a monomial
$m_n\in \M_{n+1} \setminus \M_n$.
This sequence is \emph{bad}, i.e.~$m_i \not\wpo m_j$ for every $i<j$, which is in
contradiction with $\wpo$ being a \wqo.
\end{proof}

\begin{lemma}[Correctness]\label{lem: buch correct}
Every run of Algorithm \ref{algo:buch} computes, given an input $\G\subseteq\polyring\R\WW$,
a \Gr basis of $\eqidealgen\G$.
\end{lemma}
\begin{proof}
Each polynomial $h$ added to $\G$ in the algorithm 
is obtained from some element of $\G$ by a sequence of division steps 
$\red {} \G {}$, and hence belongs to $\eqidealgen\G$. 
Therefore $\eqidealgen\G$ stays invariant during the algorithm.
Once the algorithm terminates, the set $\G$ satisfies the finite \Bu criterion and hence
is a \Gr basis, due to Lemma \ref{lem:effective buch cri}.
%
\end{proof}

\begin{slexample} \label{ex:Hilb-thm-cont}
Let $\WW = \set{a_0 < a_1 < ...}$.
Recall the ideal $\I \subseteq \polyring \Qf \WW$ from Example \ref{ex:Hilb-thm},
that contains those polynomials 
   without constant term, where the sum of coefficients of monomials of even degree 
   is equal to the sum of coefficients of monomials of odd degree.
   As we mention in Remark \ref{rem:Grob}, the base
   \begin{align}\label{eq:B}
   \B=\set{a_0^2 + a_0, a_0 - a_1}
   \end{align}
   of $\I$,
   constructed in the proof of Theorem \ref{thm:equiv noeth}, 
   is a \Gr basis. 
   Intuitively, the first polynomial in \eqref{eq:B} is sufficient to reduce all univariate polynomials from $\I$,
   while the second one is sufficient to reduce all non-univariate polynomials to univariate ones.
   
   For illustration of the operation of Algorithm \ref{algo:buch}, consider 
   some other base $\B' = \set{a_0 + a_1^2}$ of $\I$.
   We explain how Algorithm \ref{algo:buch}, given $\B'$ as its input,  
   can add to $\B'$ both polynomials from $\B$, thus computing a \Gr basis.
   As the first step, the S-polynomial (equal to the second polynomial in $\B$) 
   $$a_0 + a_2^2 - (a_1 + a_2^2) = a_0 - a_1$$ 
   is added to $\B'$, since it is reduced. Notably, an extra variable, for instance $a_2$, 
   must be used to obtain this S-polynomial.
   In particular, we learn that $\B'$ is not a \Gr base, as the newly added polynomial $a_0 - a_1$
    is reduced with respect to $\B'$.
   Then, in the second step, the
   S-polynomial 
   $$a_0 + a_1^2 + a_1(a_0 - a_1) = a_0 + a_0 a_1$$ 
   reduces 
   to the other polynomial $a_0 + a_0^2$ in $\B$,
   with respect to $\B' \cup \set{a_0-a_1}$.
   Indeed, as the remainder of division of this $S$-polynomial by $a_0-a_1$ we get 
   the polynomial
   \[
   a_0 + a_0 a_1 + a_0 (a_0-a_1) = a_0 + a_0^2,
   \]
   which is reduced and hence is added to $\B'$.
\end{slexample}

\begin{remark} \rm
In this section we strengthen the total order assumption, and assume $\WW$ to be well ordered,
so that lexicographic order is well founded.
To this effect, in Section \ref{sec:equiv noeth} we replaced leading
monomials with characteristic pairs.
In this section this is not possible, as the proof of Lemma \ref{lem:equivariant buch cri}
does not adapt.
\end{remark}

%
%



\section{Decidability of ideal membership} \label{sec:decid}
%


As an immediate consequence of Section \ref{sec:Gr} we get:
\begin{theorem} \label{thm:ideal memb decid}
\probname{Ideal-memb in $\polyring \R \WW$} is decidable,
for every computable field  $\R$ and $\omega$-well structured, well ordered and computable domain of variables $\WW$.
\end{theorem}
\begin{proof}
Given an instance $f\in\polyring\R\WW$ and $\B\subseteqfin\polyring\R\WW$ we compute,
using Algorithm \ref{algo:buch}, a \Gr basis $\G$ of $\eqidealgen\B$.
Then, relying on Lemmas \ref{lem:alg} and \ref{lem:red terminates},  
we compute some reduced polynomial $f'$ obtainable by a sequence of division steps $\redstar f \G f'$.
Relying on Lemma \ref{lem:Gr iff}, we
answer positively if $f'=\zeropol$.
\end{proof}


In the rest of this section we develop a method of transfering decidability of the ideal membership
problem to other structures $\WW$.
The method is applied to the rational order $\QQ=(\Q, \leq_\QQ)$, which is 
$\omega$-well structured and computable, but not well ordered.

\para{Reduction game}
Given a relational structure $\WW = (\W, \ldots)$ and $S,T\in \age \WW$
(source and target),
by an \emph{$(S,T)$-local embedding} we mean any partial bijection $S\to T$ (i.e.,
a bijection from a subset of $S$ to a subset of $T$) 
that extends to an embedding in $\E$.
The set of all $(S,T)$-local embeddings we call \emph{$(S,T)$-profile} of $\WW$.

\smallskip

We define a game $\game \WW {\WW'}$ induced by
two relational structures $\WW = (\W,\ldots)$ and $\WW' = (\W', \ldots)$.
The game is played by two players, \Spoiler and \Reducer, and consists of two rounds.

In the first round, \Spoiler chooses $S,T\in\age \WW$,
and \Reducer responds by $S',T' \in \age{\WW'}$.
We call $S, S'$ \emph{sources}, and $T, T'$  \emph{targets}.
Source and target may intersect.
\Reducer wins the first round if there are bijections $\sigma : S'\to S$, $\tau : T\to T'$ 
such that
$(S',T')$-profile of $\WW'$ is equal to $(S,T)$-profile of $\WW$, 
modulo $\sigma$ and $\tau$:
\begin{align} \label{eq:win1}
\prof {\WW'} {S'} {T'} \ = \ 
\setof{\tau \circ \pi \circ \sigma}{\pi \in \prof {\WW} {S} {T}}.
\end{align}
Otherwise \Spoiler wins and the play ends.

In the second round, \Spoiler chooses one of the structures $\WW$, $\WW'$,
and 
a finite superset of target in this structure.
\Reducer responds similarly in the other structure.
This step thus yields two \emph{extended targets}  $\overline T\in\age\WW$ and $\overline{T'}\in\age{\WW'}$ such that
$T\subseteq \overline T$ and $T'\subseteq \overline{T'}$.
\Reducer wins the second round (and thus the whole play) if 
$\tau$ extends to a bijection $\overline\tau : \overline T\to \overline{T'}$ 
such that:
\begin{align} \label{eq:win2}
\prof {\WW'} {S'} {\overline{T'}} \ = \ 
\setof{\overline\tau \circ \pi \circ \sigma}{\pi \in \prof {\WW} {S} {\overline T}}.
\end{align}
%
%
We write $\WW \gamered \WW'$ if \Reducer has a computable winning strategy in  
$\game {\WW} {\WW'}$.

%

\begin{lemma} \label{lem:red}
$\QQ \gamered \omega^2$.
\end{lemma}
\begin{proof}
%
A subset $X\subseteq \omega^2$ we call \emph{$k$-separated}, for $k\in\N$, if each two consecutive elements
of $X$ are separated by at least $k$ elements of $\omega^2$, and likewise there are at least
$k$ elements smaller than $\min(X)$. 

Let $S, T \in\age \QQ$ be \Spoiler's choice in the first round.
Then \Reducer's response is the following:
\begin{itemize}
\item
as source $S'$, choose
the initial fragment of $\omega$, of the same size as $S$ (therefore
$S' \subseteq \omega$);
\item
as target $T'$, choose an arbitrary 
subset of non-zero multiplicities of $\omega$, namely
$T' \subseteq \setof{\omega\cdot b}{b\in\omega, b>0}$,
of the same size as $T$.
\end{itemize}
The unique monotonic bijections $\sigma : S'\to S$ and $\tau : T \to T'$ satisfy
\eqref{eq:win1}, namely (a proof below):
\begin{align} \label{eq:win1Q}
\prof {\omega^2} {S'} {T'} \ = \ 
\setof{\tau \circ \pi \circ \sigma}{\pi \in \prof {\QQ} {S} {T}}.
\end{align}

In the second round, 
if \Spoiler chooses the structure $\omega^2$ and plays extended target
$\overline{T'} \subseteq \omega^2$, \Reducer responds by arbitrary $\overline T\in \age\QQ$ 
so that $\tau$ extends to a monotonic bijection $\overline\tau : \overline T \to \overline{T'}$.
This is always possible by denseness of $\leq_\QQ$.
Symmetrically, if \Spoiler chooses the structure $\QQ$ and plays
$\overline T \in\age \QQ$, \Reducer responds by arbitrary $\overline{T'} \subseteq \N^2$ 
so that $\tau$ again 
extends to a monotonic bijection $\overline\tau : \overline T \to \overline{T'}$.
This is always possible as $T'$ is $k$-separated for every $k\in\N$.
Furthermore, \Reducer chooses $\overline{T'}$ to be $\size{S'}$-separated.
In both cases we get \eqref{eq:win2}, namely
\begin{align} \label{eq:win2Q}
\prof {\omega^2} {S'} {\overline{T'}} \ = \ 
\setof{\overline\tau \circ \pi \circ \sigma}{\pi \in \prof {\QQ} {S} {\overline T}}.
\end{align}
Conditions \eqref{eq:win1Q} and \eqref{eq:win2Q} hold indeed, 
since $\tau$ is a monotonic bijection, and due to the following two observations:
\begin{claim}
Every monotonic partial bijection $S\to \overline T$ extends to an embedding in $\embed {\QQ}$.
\end{claim}
\begin{claim}
Every monotonic partial bijection $S'\to \overline{T'}$ extends to an embedding in $\embed {\omega^2}$.
\end{claim}
\noindent
The former claim follows by homogeneity of $\QQ$.
For the latter claim, 
we first extend the mapping to $S'$, which is possible as $\overline{T'}$ is $\size{S'}$-separated, and then
define an extension on remaining arguments $\a\in\N^2\setminus S'$ by $\a\mapsto \a+\b$,
for a sufficiently large ordinal $\b$.
\end{proof}

Corollary \ref{cor:red d} below may be shown in the same way as Lemma \ref{lem:red},
but it also follows from the general fact that the product of domains preserves the winner 
in the reduction game:
\begin{lemmarep} \label{lem:red prod}
If $\WW_1\,{\gamered}\, \WW'_1$ and $\WW_2\gamered\WW'_2$ then
$\WW_1 {\domprod} \WW_2 \gamered \WW'_1 {\domprod} \WW'_2$.
\end{lemmarep}
\begin{proof}
Assuming  \Reducer's strategies in
games $\game {\WW_1} {\WW'_1}$ and $\game {\WW_2} {\WW'_2}$, we need to provide his strategy in
 $\game {\WW_1\domprod\WW_2} {\WW'_1\domprod\WW'_2}$.
Suppose \Spoiler plays in the first round a source and target 
\[
S, T\in\age{\WW_1\domprod \WW_2}.
\]
Let $S_1, T_1 \in \age {\WW_1}$ and $S_2, T_2\in\age{\WW_2}$ be the projections.
From \Reducer's responses in $\game {\WW_1} {\WW'_1}$ and
$\game {\WW_2} {\WW'_2}$, we get
$S'_1, T'_1 \in \age {\WW'_1}$, $S'_2, T'_2\in\age{\WW'_2}$  and bijections
\begin{align*}
\sigma_1 : S'_1 \to S_1& \qquad \tau_1 :  T_1 \to T'_1 \\
\sigma_2 : S'_2 \to S_2& \qquad \tau_2 :  T_2 \to T'_2.
\end{align*}
The products of these bijections,
\begin{align*}
(\sigma_1, \sigma_2) : S'_1\times S'_2 \to S_1\times S_2 \qquad 
(\tau_1, \tau_2) :  T_1\times T_2 \to T'_1\times T'_2,
\end{align*}
are used in \Reducer's response 
$S', T'\in\age{\WW'_1\domprod \WW'_2}$:
\[
S' \ := \ (\sigma_1, \sigma_2)^{-1}(S) \qquad
T' \ := \ (\tau_1, \tau_2)(T).
\]
The restrictions of $(\sigma_1, \sigma_2)$ and $(\tau_1, \tau_2)$ to $S'$ and $T$, respectively,
\[
\sigma = \restr{(\sigma_1, \sigma_2)}{S'} \qquad
\tau = \restr{(\tau_1, \tau_2)}{T},
\]
satisfy \eqref{eq:win1}. Indeed, condition \eqref{eq:win1} follows directly by Lemma \ref{lem:prod embed}, and by 
the same condition for games $\game {\WW_1} {\WW'_1}$ and $\game {\WW_2} {\WW'_2}$.

Similarly we define \Reducer's strategy in the second round of $\game {\WW_1\domprod\WW_2} {\WW'_1\domprod\WW'_2}$.
\end{proof}
\begin{corollary} \label{cor:red d}
$\QQ \domprod \finitedomain d \gamered \omega^2 \domprod \finitedomain d$.
\end{corollary}

\begin{theorem}\label{thm:red ord to well-ord}
Let $\R$ be a computable field and $\WW, \WW'$ two relational structures.
If $\WW \gamered \WW'$ then
\probname{Ideal-memb in $\polyring \R \WW$} reduces to
\probname{Ideal-memb in $\polyring \R {\WW'}$}.
\end{theorem}
\begin{proof}
Given an instance of \probname{Ideal-memb in $\polyring \R \WW$}, i.e., 
a polynomial $f$ and a finite basis $\B$,
we compute an instance $f'$, $\B'$ of  \probname{Ideal-memb in $\polyring \R {\WW'}$} 
relying on \Reducer's strategy in the first round of $\game \WW {\WW'}$.
We let \Spoiler play $S:=\vars \B$, the set of variables appearing in $\B$, and
$T:=\vars f$.
From the \Reducer's answer we get two subsets $S', T'\subseteq \W'$ and bijections
$\sigma : S'\to S$ and $\tau : T\to T'$.
We use this bijections to rename variables in $f$ and $\B$, thus obtaining $f'$ and $\B'$.

Correctness of the reduction is shown using \Reducer's strategy in the second round.
In one direction,
suppose $f$ belongs to the ideal generated by $\B$ in $\polyring \R \WW$, i.e.,
\begin{align} \label{eq:ideal witness}
f \ = \ h_1 \cdot \iota_1(g_1) \ + \ \ldots \ + \ h_m \cdot \iota_m(g_m),
\end{align}
for polynomials $g_1, \ldots, g_m \in \B$, polynomials $h_1, \ldots, h_m \in \polyring \R \WW$ and
embeddings $\iota_1, \ldots, \iota_m \in E$.
We let \Spoiler play $\overline T$ containing all variables appearing on the right in equality \eqref{eq:ideal witness}.
Hence $T \subseteq \overline T$.
From \Reducer's response we get  $\overline{T'}$ and a bijection
$\overline\tau : \overline T \to \overline{T'}$ that extends $\tau$. 
We use $\sigma$ and $\overline\tau$ to rename all variables 
in polynomials $h_1, \ldots, h_m$ and $g_1, \ldots, g_m$ appearing in the decomposition 
\eqref{eq:ideal witness}, which
results in polynomials $h'_1, \ldots, h'_m \in \polyring \R {\WW'}$ and
$g'_1, \ldots, g'_m \in \B'$.
Furthermore, for every $i \in\setfromto 1 m$, we apply
the equalities \eqref{eq:win1}, \eqref{eq:win2} 
to the restriction 
$\pi_i = \restr {\iota_i} S$ to deduce that
$\overline\tau \circ \pi_i \circ \sigma$ extends to some embedding
$\iota'_i \in E'$.
Therefore, we get:
\begin{align} \label{eq:ideal witness'}
f' \ = \ h'_1 \cdot \iota'_1(g'_1) \ + \ \ldots \ + \ h'_m \cdot \iota'_m(g'_m).
\end{align}
The reverse direction is shown similarly: assuming a decomposition \eqref{eq:ideal witness'},
we let \Spoiler choose $\WW'$ in the second round, and
use the bijections $\sigma$ and $\tau$ obtained from \Reducer's response 
to get a decomposition \eqref{eq:ideal witness}.
\end{proof}
\vspace{-1.5mm}
Using Lemmas \ref{lem:prod}, \ref{lem:lex comp}, and \ref{lem:red prod}, 
Theorems \ref{thm:ideal memb decid},
\ref{thm:red ord to well-ord},  
and Corollary \ref{cor:red d} we get reduction
of \probname{Ideal-memb in $\polyring \R {\QQ\domprod \finitedomain d}$} to
\probname{Ideal-memb in $\polyring \R {\omega^2 \domprod \finitedomain d}$}.
%
%
We easily generalize this reduction by invoking Lemma \ref{lem:red prod} several times.
For conciseness, we write $\domexp \WW \ell$ for the $\ell$-fold product:
$
\domexp \WW \ell \ := \  \underbrace{\WW \domprod \ldots \domprod  \WW}_\ell.
$

\begin{corollary}\label{cor:red times d l}
For every $d,\ell\geq 1$,
\probname{Ideal-memb in $\polyring \R {\domexp \QQ \ell \domprod \finitedomain d}$} reduces to
\probname{Ideal-memb in $\polyring \R {\domexp{(\omega^2)} \ell \domprod \finitedomain d}$}
and is therefore decidable.
\end{corollary}

%


\section{Application: reversible Petri nets with data} \label{sec:pn}

In this section we show that Theorem \ref{thm:ideal memb decid} is useful for showing
decidability of the reachability problem for \emph{reversible} Petri nets whose tokens carry data
from a given domain
(Theorem \ref{thm:red reach}).

\para{Petri nets with data}
Among many presentations of Petri nets, we choose multiset rewriting, except that in place of
finite multisets over $\WW$ we conveniently take monomials $\monomials \WW$.

Let $\WW = (\W, \ldots)$ be a fixed relational structure.
Any set of rewriting rules of the form
\[
\transitions \subseteq \monomials \WW \times \monomials \WW
\]
we call \emph{monomial rewriting over $\WW$}.
Rewriting rules extend to \emph{transitions} $\trans \ \subseteq \monomials \WW \times \monomials \WW$ as follows:
for every $g\in\monomials \WW$ and $(h, h') \in \transitions$, we have a transition
\[
g \cdot h \trans g \cdot h'.
\]
The \emph{reachability} relation is the transitive closure $\trans^*$.
In the sequel we always assume that $\transitions$ is orbit-finite, namely
\[
\transitions \ = \ \E(T) \ \defeq \ \setof{(\iota(h), \iota(h'))}{\iota \in \E, \ (h,h')\in T}
\]
for some finite set $T\subseteqfin \monomials \WW \times \monomials \WW$.
The set $T$ serves as a \emph{presentation} of $\transitions$ when it is input to algorithms.
A \emph{Petri net with data $\WW$} is just a monomial rewriting over
$\WW\domprod \finitedomain d$ for some $d\in\N$:

\begin{definition}
A Petri net with data $\WW$, with $d$ places, is an orbit-finite monomial rewriting $\transitions$ 
over $\WW\domprod \finitedomain d$.
\end{definition}

\begin{slexample}
In particular, monomial rewriting over $\finitedomain d$ defines a plain (data-less) Petri net with 
$d$ places $\setto {d-1}$. 
The only embedding is identity, and hence $\transitions$ is finite.
Monomial rewriting over $\QQ\domprod \finitedomain d$,
where $\QQ=(\Q, \leq_\QQ)$ is the rational order,  defines a Petri net with ordered data.
According to Lemma \ref{lem:prod embed},
embeddings of $\QQ \domprod \finitedomain d$ are induced by some embedding 
$\iota$ of $\QQ$ (i.e., a strictly increasing function $\Q\to\Q$), 
and preserve the second component, i.e., preserve places
(for conciseness, instead of $(\w,i)$ we write $\w i$):
\[
\w i \mapsto \iota(\w) i.
\]
%
Classical graphical presentation of Petri nets is extendable to Petri nets with data.
For instance, consider a Petri net $\transitions = \embed \QQ(T)$ with $d=2$ places, 
and with 2 transition rules $T = \set{t_1, t_2}$, where
\[
t_1 = (\zeromon, \ \w0 \cdot \v0) \qquad
t_2 = (\v0 \cdot \v1, \ \w0 \cdot \ww1)
\]
for some $a<b<c\in\Q$ ($\zeromon$ is the unit monomial).
Petri net $\transitions$ is presented symbolically,
using \emph{variables} $\w,\v,\v', \ww$ ranging over $\Q$:

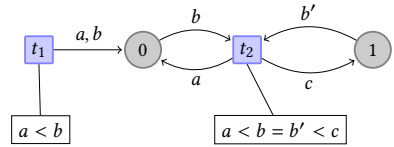
\begin{center}
\scalebox{.85}{
\begin{tikzpicture}
[auto,place/.style={circle,draw=black!50,fill=black!20,thick,inner sep=2mm},
transition/.style={rectangle,draw=blue!50,fill=blue!20,thick,rounded corners=0.5pt,inner sep=1mm}]
 
\node (P) [place] {}
  [children are tokens]
  child{node{0}};
\node (Q) [place,right=\picunit of P] {}
  [children are tokens]
  child{node{1}};
\node (T) [transition, left=\halfunit of P,label=above:\ ] {$t_1$}
  edge[post] node{$\w,\v$}(P);
\node (S) [transition, right=\halfunit of P,label=above:\ ] {$t_2$}
  edge[pre,bend right,above] node{$\v$}(P)
  edge[pre,bend left,above] node{$\v'$}(Q)
  edge[post,bend left,below] node{$\w$}(P)
  edge[post,bend right,below] node{$\ww$}(Q);
\node[rectangle,draw] at (-1.60 ,-1.25) (phi1) {$\w < \v$};
\draw  [-] (T.south)      -- (phi1);
\node[rectangle,draw] at (2.15 ,-1.25) (phi2) {$\w < \v = \v' < \ww$};
\draw  [-] (S.south)      -- (phi2);
\end{tikzpicture}
}
\end{center}

Transition rule $t_1$ outputs two tokens with arbitrary but distinct rational values onto place $0$. 
Transition $t_2$ inputs two tokens with the same value, say $\v$, one from $0$ and one from $1$, and outputs two tokens:  one token $\w<\v$ onto $0$ and the other $\ww>\v$ onto $1$.
\end{slexample}

$\transitions$ is \emph{reversible} if it is closed under inverse:
$c\trans^* c'$ implies $c' \trans^* c$ for every $c,c'\in\monomials \WW$.
In such case we can assume, w.l.o.g., that $T$ is closed under inverse:
$(h, h')\in T \implies (h',h) \in T$.

\para{Reachability problem}
The \emph{reachability problem} asks, given a Petri net $\transitions$ 
and two monomials, source and target $s, t\in\monomials {\WW\domprod \finitedomain d}$,
whether $s \trans^* t$.

\begin{theorem}\label{thm:red reach}
For every $d\geq 1$,
the reachability problem in reversible Petri nets with $d$ places,
with data $\WW$, 
reduces to
\probname{Ideal-memb in $\polyring \Qf {\WW\domprod \finitedomain d}$}
(and also to
\probname{Ideal-memb in $\polyring \Z {\WW\domprod \finitedomain d}$}).
\end{theorem}
%
%
\begin{proof}
The fundamental idea of the proof comes from \cite{MayrMeyer82}.
Let $\transitions = \embed \WW(T)$ be a Petri net with data $\WW$, with $d$ places,
where
$T=\set{(\c_1, \d_1), \ldots, (\c_n, \d_n)}$.
Let $s,t\in\monomials  {\WW\domprod \finitedomain d}$ be source and target monomial.
The reduction transformes $s, t$ and $T$ to
the binomial
$(t-s)\in\polyring \Z  {\WW\domprod \finitedomain d}$ and the finite set of binomials
$\B=\set{(\d_1 - \c_1), \ldots, (\d_n - \c_n)} \subseteq \polyring \Z  {\WW\domprod \finitedomain d}$.
Correctness follows due to equivalence
of the following three conditions:
\begin{enumerate}
\item $\c\trans^*\d$;
\item $(\d-\c)$ belongs to the equivariant ideal generated by $\B$ in $\polyring \Z  {\WW\domprod \finitedomain d}$;
\item $(\d-\c)$ belongs to the equivariant ideal generated by $\B$ in $\polyring \Qf  {\WW\domprod \finitedomain d}$;
\end{enumerate}
The implications $(1){\implies} (2)$ and $(3) {\implies} (1)$ are shown by 
reusing the proofs of Lemmas 1 and 2 in \cite{MayrMeyer82}. 
The remaining implication $(2){\implies} (3)$ is trivial.
\end{proof}

Reduction of Theorem 
\ref{thm:red reach} is uniform in the number $d$ of places. 
So is also the construction of \Gr basis underlying the proof of
Theorem \ref{thm:ideal memb decid}, when 
instantiated to domains of the form $\WW \domprod \finitedomain d$.
We thus deduce decidability of 
the reachability problem for reversible Petri nets  for a wide range of 
domains $\WW$, including $\QQ$, equality data
(a sub-case of $\QQ$), and lexicographic
products 
thereof (on the basis of Corollary \ref{cor:red times d l}):
\begin{theorem}\label{thm:red reach l}
For every $\ell\geq 1$,
the reachability problem is decidable in reversible Petri nets with 
$\ell$-fold nested ordered data $\domexp \QQ \ell$. 
\end{theorem}
\begin{remark} \rm
Without reversiblity restriction, the reachability problem for Petri nets with ordered data 
is undecidable, while for Petri nets with equality data its status is still unknown.
\end{remark}


\section{Application: Orbit-finite systems of linear equations} 
\label{sec:eq}

We demonstrate that
Theorem \ref{thm:ideal memb decid} is applicable to solving orbit-finite systems of linear equations,
thus providing a simple proof of a generalisation of the main result of \cite{GHL22}.
In the following, $\R$ is a field and $\WW = (\W, \ldots)$ an $\omega$-well structured 
relational structure.


Solving a system of linear equations amount to checking if a given target vector $b$ 
is a linear combination of a given set $M$ of vectors.
We generalise this setting to orbit-finite sets, by assuming that $b\in \lin{\otu \W d}$
and that $M$ is an equivariant orbit-finite family of vectors $M\subseteq\lin{\otu \W d}$.
For convenience we assume that the family $M$  
is \emph{indexed} by a set $\otu \W n$ for some $n\in\N$, 
namely 
\[
M = (M_t\in\lin{\otu \W d})_{t \in\otu \W n}
\]
(this is w.l.o.g.~for equality and rational order domain \cite[Sect.6]{GHL22}).
Such a family $M$ is equivariant if for every $\iota\in\E$ and $t\in\otu \W n$, 
\[
\iota (M_t) = M_{\iota(t)},
\]
i.e., $M_\_$ commutes with the action of embeddings, as the function
$\otu \W n \to \lin{\otu \W d}$.
Alternatively, $M$ is an infinite $\otu \W d \times \otu \W n$ matrix
over $\R$,
whose every column $M_t$, being a vector in $\lin{\otu \W d}$, 
has finitely many non-zero entries.
A linear combination of $M$ is of the form
\begin{align} \label{eq:lincomb}
r_1 \cdot M_{t_1} + \ldots + r_k \cdot M_{t_k},
\end{align}
where $k\in\N$, $t_1, \ldots, t_k \in \otu \W n$, and $r_1, \ldots, r_k\in\R$.
We investigate the following problem:

\smallskip
\begin{description}                                                                                                                                                               
\item[\probname{Lin-solv($\R, \WW$)}:]           
\item[\bf Input: ] \ \ \ \ \ \ \ a vector $b\in\lin{\otu \W d}$ and an equivariant family
\item[] \ \ \ \ \ \ \ \ \ \ \ \ \ \ \ \  $M = (M_t\in \lin{\otu \W d})_{t\in\otu \W n}$ of vectors.
\item[\bf Question: ] Is $b$ a linear combination of $M$?
\end{description}                                                                                                                                                                 
%
%
\begin{slexample}
Let $\R=\Rf$, $\WW = (\W,=)$ the equality domain,
$d=1$ and $n=2$, and consider
an equivariant family of vectors $M = (M_{(a,b)} \in \lin{\W})_{(a,b)\in\W^2}$
defined as the following formal sums:
\[
M_{(a,b)} =
\begin{cases}
a + 2\cdot b & \text{ when } a \neq b \\
0 & \text{ when } a = b \ .
\end{cases}
\]
For every $\w\in\W$, the singleton formal sum $\w\in\lin{\W}$ 
is a linear combination of $M$:
\[
\w
\ = \ \frac 1 2 \cdot M_{(\w,\v)} \ + \ \frac 1 2 \cdot M_{(\w,\ww)} \ - \ 
\frac 1 3 \cdot M_{(\v,\ww)} \ - \ \frac 1 3 \cdot M_{(\ww,\v)},
\]
for any distinct $\v, \ww\in \W$ different than $\w$.
\end{slexample}
%
%
\begin{theorem} \label{thm:eq}
\probname{Lin-solv($\R, \WW$)} 
reduces to
\probname{Ideal-memb in $\polyring \R \WW$}, 
for every computable field $\R$ and relational structure $\WW$.
\end{theorem}
Theorem \ref{thm:eq} yields decidability of \probname{Lin-solv($\R, \WW$)}, 
thus generalising \cite[Theorem 4.4]{GHL22} which only considers
equality structure. 
On the other hand, \cite{GHL22} provides an upper complexity bound.
%
\begin{proof}[Proof of Theorem \ref{thm:eq}]
The reduction uses the same encoding $\enc \_ : \lin{\otu \W d} \to \polyring \R \WW$ as in 
Section \ref{sec:lin}.
Consider an instance of \probname{Lin-solv($\R, \WW$)}, namely 
$M = (M_t)_{t\in\otu \W n}$ and $b$.
Relying on Lemma \ref{lem:of} we assume that 
$\otu \W n$ is represented by a finite set $\Bs\subseteqfin \otu \W n$ such that
$\W^n = \E(\Bs) \defeq \setof{\iota(t)}{\iota\in\E, t\in\Bs}$.

We construct an instance of \probname{Ideal-memb in $\polyring \R \WW$} by
taking the polynomial $f = \enc{b}\in\polyring \R \WW$, and the equivariant $\I$ ideal generated
by the finite set
\[
\setof{\enc {M_t}}{t\in\Bs} \ \subseteq \ \polyring \R \WW.
\]
As encoding commutes with the action of embeddings,
$\I$ is generated classically by the orbit-finite set
$\setof{\enc {M_t}}{t\in\W^n}$.
We need to prove:
\begin{claim} \label{claim:enc-corr}
$b$ is a linear combination of $M$ if and only if $f\in \I$.
\end{claim}
The 'only if' direction is immediate: $\enc \_$, as a linear function, maps a linear combination
of vectors to a linear combination of their encoding polynomials. 
For the converse implication, suppose 
$f = \enc b$ belongs to $\I$, i.e.,
\[
\enc b  \ = \  r_1 \cdot h_1 \cdot \iota_1(\enc{M_{t_1}})  +  \ldots  + 
  r_1 \cdot h_k \cdot \iota_k(\enc{M_{t_k}}) 
\]
for some $r_1, \ldots, r_k\in\R$, $h_1, \ldots, h_k \in \monomials \WW$,
$\iota_1, \ldots, \iota_k \in \E$ and $t_1, \ldots, t_k\in \Bs$.
As all monomials appearing in polynomials $\enc{M_{t_i}}$ and $\enc b$ have the same degree,
similarly as in the proof of Theorem \ref{thm:lin}  we can assume, w.l.o.g., 
that all monomials
$h_i$ are unit monomials, i.e.,
\[
\enc b  \ = \  r_1 \cdot \iota_1(\enc{M_{t_1}})  +  \ldots  +  r_1 \cdot \iota_k(\enc{M_{t_k}}).
\]
As encoding commutes with the action of embeddings and with scalar multiplication, we get
\[
\enc b  \ = \  \enc{r_1 \cdot \iota_1(M_{t_1})}  +  \ldots  +  \enc{r_1 \cdot \iota_k(M_{t_k})},
\]
and by injectivity of encoding and equivariance of $M$ we obtain 
\[
 b  \ = \  r_1 \cdot M_{\iota_1(t_1)}  +  \ldots  +  r_1 \cdot M_{\iota_k(t_k)},
\]
i.e.~$b$ is a linear combination of $M$, as required.
\end{proof}


\begin{acks}
We thank Lorenzo Clemente for posing the questions about
equivariant Noetherianity for equality domain.
\end{acks}

\newpage


\bibliographystyle{ACM-Reference-Format}
\bibliography{bib}


\appendix


\end{document}